\documentclass[lettersize,journal]{IEEEtran}
\usepackage{amsmath,amsfonts}
\usepackage{amsthm}
\usepackage{algorithm}
\usepackage{algpseudocode}
\usepackage{array}
\usepackage{textcomp}
\usepackage{stfloats}
\usepackage{url}
\usepackage{verbatim}
\usepackage{graphicx}
\usepackage{cite}
\usepackage[flushleft]{threeparttable}
\usepackage{multirow}
\usepackage{subcaption}
\def\BibTeX{{\rm B\kern-.05em{\sc i\kern-.025em b}\kern-.08em
    T\kern-.1667em\lower.7ex\hbox{E}\kern-.125emX}}
\usepackage{balance}
\newtheorem{proposition}{Proposition}

\usepackage{color}
\begin{document}

\title{Partner in Crime: Boosting Targeted Poisoning Attacks against Federated Learning}


\author{\IEEEauthorblockN{Shihua Sun,
Shridatt Sugrim,
Angelos Stavrou,
Haining Wang} 
\thanks{Shihua Sun, Angelos Stavrou, and Haining Wang are with the Department of Electrical and Computer Engineering, Virginia Tech, Arlington, VA 22203 USA (e-mail: shihuas@vt.edu; angelos@vt.edu; hnw@vt.edu).}
\thanks{Shridatt Sugrim is with Kryptowire Labs, McLean, VA 22102 USA (e-mail: ssugrim@kryptowire.com).}
}



\maketitle

\begin{abstract}
Federated Learning (FL) exposes vulnerabilities to targeted poisoning attacks that aim to cause misclassification specifically from the source class to the target class. However, using well-established defense frameworks, the poisoning impact of these attacks can be greatly mitigated. We introduce a generalized pre-training stage approach to {\it Bo}ost {\it T}argeted {\it P}oisoning {\it A}ttacks against FL, called BoTPA.
Its design rationale is to leverage the model update contributions of all data points, including ones outside of the source and target classes, to construct an Amplifier set, in which we falsify the data labels before the FL training process, as a means to boost attacks. 
We comprehensively evaluate the effectiveness and compatibility of BoTPA on various targeted poisoning attacks. Under data poisoning attacks, our evaluations reveal that BoTPA can achieve a median Relative Increase in Attack Success Rate (RI-ASR) between 15.3\% and 36.9\% across all possible source-target class combinations, with varying percentages of malicious clients, compared to its baseline.
In the context of model poisoning, BoTPA attains RI-ASRs ranging from 13.3\% to 94.7\% in the presence of the Krum and Multi-Krum defenses, from 2.6\% to 49.2\% under the Median defense, and from 2.9\% to 63.5\% under the Flame defense.

\end{abstract}


\section{Introduction}

Federated Learning (FL)~\cite{fl_2,fl_1} has recently emerged as a scalable approach to generate Machine Learning (ML) models without the need to share privately owned data. Indeed, with FL, data can be kept in disjoint local sets, inherently offering better privacy for sensitive information. This feature makes FL a promising framework for a multitude of applications with dense device deployments at the edge, such as resource allocation for low-latency vehicle-to-vehicle (V2V) communications~\cite{FL_V2V}, mobile keyboard prediction by Google~\cite{FL_keyboard}, intrusion detection in IoT networks~\cite{fl_IDS,fedmade} and patient similarity analysis in healthcare~\cite{FL_healthcare,patient_hashing}.  
However, despite its performance scalability and intrinsic privacy benefits, FL remains susceptible to numerous privacy~\cite{model_inversion_attack,unintended_leakage,mermbership_FL,scale-MIA,gradient_leakage,privacy_fl} and security attacks~\cite{adv_len,local_poison,FL_datapoisoning}. This is primarily due to the distributed nature of FL, which increases the exposed attack surfaces through malicious FL participants. 

Conceptually, poisoning attacks attempt to manipulate the dataset or the model update process for misclassifications, eroding the trustworthiness of FL. Poisoning attacks can be broadly classified into targeted~\cite{FL_backdoor1,adv_len,FL_datapoisoning,FL_backdoor2} and untargeted poisoning attacks~\cite{untargeted_poisoning_1,untargeted_poisoning_2,MPAF,untargeted_poisoning_3}. The former seeks to deliberately misclassify data points from specific source classes to target classes, while the latter aims to deteriorate the overall performance of the ML model. Both attack variants can be accomplished through data poisoning~\cite{FL_datapoisoning,data_poison} or model poisoning~\cite{adv_len,MPAF,local_poison,optimizing_model_poisoning} techniques. In data poisoning attacks, adversaries can only manipulate the data samples and cannot directly control the model training process. In contrast, during model poisoning attacks, adversaries can have a significant impact on the model by modifying the model parameters and the loss function, for instance. 

The focus of this paper is on targeted poisoning attacks, an area that has been less explored compared to untargeted attacks in prior research. 
In practice, the strength and impact of existing poisoning attacks~\cite{adv_len,untargeted_poisoning_2,FL_datapoisoning} are limited by the number of subverted FL clients, the presence of defenses~\cite{KRUM,fltrust,Bulyan,median}, and the attack capability. In data poisoning attacks, adversaries can only affect the model updates by contaminating the local data or inserting poisonous data. This limits the ability of adversaries to affect the training process.
In model poisoning attacks, adversaries are allowed to modify the model weights and the training process. This is much more powerful when compared to data poisoning attacks. However, the divergence between malicious and benign updates can be readily detected by well-established defenses~\cite{KRUM,median}.
To evade these defenses, researchers introduced stealthy model poisoning attacks~\cite{adv_len}, which offer a trade-off between remaining undetected and improving attack success rates (ASRs). Unfortunately, being stealthy comes at a high cost, and in many cases, the ASRs of poisoning attacks can be significantly reduced~\cite{adv_len}. Furthermore, the majority of untargeted poisoning attacks are not applicable in targeted scenarios due to the distinct objectives and attack strategies involved. 

Prior research~\cite{adv_len,poison_frogs,FL_datapoisoning} has focused exclusively on the direct relationship between the source and target classes when executing targeted poisoning attacks. The most intuitive strategy~\cite{FL_datapoisoning} in these scenarios involves mislabeling data samples from the source class as the target class. However, there are multiple factors affecting the decision boundary, with data distribution playing a significant role~\cite{smote,imbalance}. Notably, an important observation is that poorly separated classes in the feature space result in less distinct decision boundaries. Exploiting this insight, an attacker can deliberately manipulate the feature space to increase overlap by modifying the source class and other classes.
This raises two critical research questions: (1) how to identify the additional classes that significantly affect the decision boundary between the source and target classes, and (2) how to mislabel these classes to enhance the poisoning effect while minimizing the likelihood of detection.


In this paper, we propose a general pre-training approach to boost the ASRs of existing targeted poisoning attacks without affecting the stealthiness of original (vanilla) attacks. As a framework for \textit{Bo}osting \textit{T}argeted \textit{P}oisoning \textit{A}ttacks, BoTPA is able to breach the boundary between the source and target classes in the latent feature space. More specifically, BoTPA selects specific classes excluding the source and target classes, designated as \textbf{intermediate classes}, based on their model update contributions to the poisoned training process. An \textbf{Amplifier} set is then constructed using the existing samples from these intermediate classes. 
This design allows BoTPA to enhance the attack's effectiveness without relying solely on direct source-target modifications.
Subsequently, BoTPA intelligently crafts malicious soft labels to replace normal data labels in the Amplifier set. In contrast to one-hot encoding labels (hard labels), soft labels~\cite{pen_vis,rethink_label_smoothing} composed of fractional numbers have a higher entropy because there are more values that the label vector can take on. This increase in the density of values makes subtle modifications of the data labels much harder to detect. The total label modification mass can be more easily spread over the label components while ensuring a bound on the individual component deltas. 
To the best of our knowledge, ours is the first work to leverage intermediate classes for launching a poisoning attack.

\textbf{Contributions:} We present the design, implementation and evaluation of BoTPA against FL models. BoTPA involves two critical steps: (1) selecting intermediate classes that generate model updates similar to those of the source class during the poisoned training process and subsequently forming the Amplifier set; (2) crafting malicious soft labels for the samples within the Amplifier set. 
In the initial stage of BoTPA, the attacker trains a surrogate model using datasets from malicious clients.
By collectively considering samples from each class, the attacker employs the \textit{Class Similarity from Contribution Degrees}~\cite{input_similarity} metric to identify the intermediate classes. Subsequently, the Amplifier set is formed by incorporating data samples from these intermediate classes.
Next, based on the \textit{Class Similarity from Latent Feature Distributions}, the attacker designs soft labels for the Amplifier set to enhance the local poisoning impact. 
Finally, within the malicious clients, the data labels in the Amplifier set are modified to the crafted soft labels. 

We apply BoTPA to data poisoning attacks~\cite{FL_datapoisoning} and stealthy model poisoning attacks~\cite{adv_len}. Throughout the rest of this paper, unless specified otherwise, the term ``poisoning attack'' refers to targeted poisoning attacks. Our experiments are conducted across three datasets: FMNIST, CIFAR-10, and CH-MNIST, with a range of malicious client percentages from 5\% to 30\%.
The experimental results demonstrate that BoTPA can significantly enhance the effectiveness of targeted poisoning attacks while maintaining stealthiness under the Byzantine-resilient defenses Krum~\cite{KRUM} and Median~\cite{median}, as well as the state-of-the-art defense Flame~\cite{flame}.
The summary of our main contributions is given below:
\begin{itemize} 
    \item We propose BoTPA, a pre-training stage boosting technique compatible with various types of targeted poisoning attacks, including both data and model poisoning attacks. 
    BoTPA requires only modifications to existing data labels on malicious clients, without the insertion of additional malicious data samples.
     
    \item We provide a theoretical analysis demonstrating that the direction of malicious model updates augmented by BoTPA aligns with that of malicious updates under vanilla attacks, thereby enhancing the targeted poisoning impact.

    \item We extensively evaluate the effectiveness and stealthiness of BoTPA in both independent and identically distributed (IID) and non-IID data distribution scenarios.

    \item We present visualizations of feature distributions in the latent space to elucidate how BoTPA alters the decision boundary between the source and target classes.


\end{itemize}
\section{Background and Related Work}
In this section, we present the fundamentals of the FL system, followed by an exploration of how it can be compromised by poisoning attacks, as well as the defenses against them.

\subsection{FL System}
We consider an FL system comprising a global server and $K$ local clients, denoted by $\{\mathcal{C}_1, \mathcal{C}_2, \cdots, \mathcal{C}_K\}$.  In this system, the local clients collaborate to develop a well-trained model while keeping their data private from both the server and other participants. 
All training operations are performed locally by the clients, while the server is responsible for aggregating the local model updates.
The workflow of the FL system is outlined as follows: (1) The server initializes a global model $\boldsymbol{w}_{0}$ and broadcasts to clients. (2) Each local client computes the model update $\boldsymbol{\delta}_{t}^{(k)}$ and sends it back to the server. (3) The server applies the aggregation rule to the received updates. With FedAvg~\cite{FedAVG}, the global update is $\boldsymbol{\delta}_{t}=\sum_{k=1}^{K} \frac{n^{(k)}}{\sum_{k=1}^{K} n^{(k)}} \boldsymbol{\delta}_{t}^{(k)}$, where $n^{(k)}$ is the size of the dataset in the $k$-th client. The updated global model is then computed as $\boldsymbol{w}_{t} = \boldsymbol{w}_{t-1} + \boldsymbol{\delta}_{t}$. (4) The server broadcasts the updated global model to the clients. Steps (2) to (4) are repeated iteratively until the model converges.

\begin{figure}[t]
\centerline{\includegraphics[width=\linewidth, trim =0.5cm 8.2cm 12.3cm 1.5cm, clip]{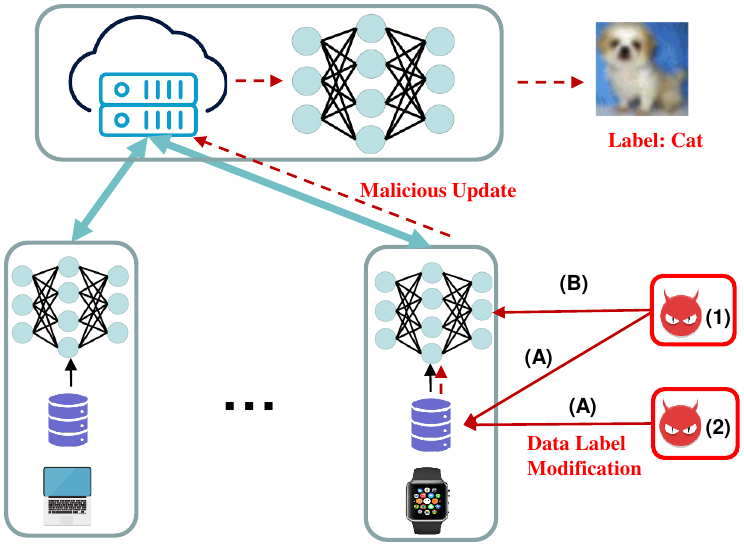}}
\caption{The FL system under targeted data poisoning and model poisoning attacks: (A) procedure for maliciously changing labels, (B) model manipulation process. Attacker (1) is performing model poisoning, and Attacker (2) is performing data poisoning.}
\label{FL_system}
\end{figure}

\subsection{Untargeted Poisoning Attacks in FL systems}

 Tolpegin et al.~\cite{FL_datapoisoning} introduced label-flipping attacks, which modify only training data labels. Bhagoji et al.~\cite{adv_len} proposed model poisoning attacks that employ an alternative minimization approach to maintain stealthiness. Fang et al.~\cite{local_poison} formulated the attack as an optimization problem to craft malicious updates in the opposite direction of the normal global model update without an attack. Zhang et al.~\cite{poisonGAN} presented PoisonGAN, a method based on generative adversarial networks to generate fake data samples resembling the data distribution of malicious clients. A more stealthy attack approach is backdoor attacks\cite{boba,FL_backdoor1,DP_backdoor1,FL_backdoor2}, where misclassifications are induced through triggers attached to data features. However, our boosting method focuses solely on label manipulation without feature modification. Thus, backdoor attacks with triggers are beyond the scope of this paper.

Most of the untargeted attack approaches cannot be adapted or extended to targeted poisoning attacks and, as a result, are not applicable to the problem domain we focus on. 

\subsection{Targeted Poisoning Attacks in FL systems} \label{sec:background_targeted}
Here, we outline the strategies for both targeted data poisoning~\cite{FL_datapoisoning} and targeted model poisoning~\cite{adv_len} attacks, which serve as baseline adversarial methods in this study and are collectively referred to as vanilla attacks. 
Figure~\ref{FL_system} provides an overview of the FL system under these poisoning attacks.
Specifically, attacker 1 performs model poisoning and 
 attacker 2 conducts data poisoning. Both aim to degrade the classification performance of the target model. While both attackers can manipulate the training data, the model poisoning attacker can further alter the FL model directly, leading to more severe performance degradation.
A detailed discussion of these attack strategies is given below.

\subsubsection{Targeted data poisoning attack}     
Before training, the attacker maliciously modifies the data labels from the source class to the target class~\cite{FL_datapoisoning}. The attacker is not allowed to directly alter the model weights or the loss function.

\subsubsection{Targeted model poisoning attacks}
Two types of targeted model poisoning attacks are introduced below. It is important to note that both attacks begin by altering the source labels to match the target labels.
(1) Model poisoning attack by explicitly boosting updates~\cite{adv_len}: Instead of sending the normal updates $\delta^{(k)}$ back, malicious clients apply a boosting factor $\lambda$ to them. This approach aims to counteract the downscaling effect caused by the FedAvg aggregation method by reporting $\lambda \delta^{(k)}$. 
(2) Stealthy model poisoning attack by alternating minimization~\cite{adv_len}: This strategy aims to improve the attack's stealthiness by considering two detection metrics—accuracy on validation data and weight update statistics. Accordingly, two specific loss terms are incorporated into the poisoning objective function: training loss over the benign dataset and the distance between malicious and benign updates.

In addition to the attacks mentioned above, there are other existing targeted attacks for centralized ML models or FL systems. However, these methods are not appropriate for baseline comparisons with our boosting method due to different attack assumptions and settings.
Shafahi et al.~\cite{poison_frogs} introduced an optimization-based method for crafting poisoned samples, which are subsequently injected into the training dataset. In contrast, our boosting mechanism only requires label manipulation without data injection, making it a weaker assumption for the adversary.
 Jagielski et al.~\cite{subpopulation} designed a subpopulation selection method and then generated attacks using a subset of the training data. However, their approach assumes access to large and diverse datasets, which does not align with the distributed setting of FL systems. Wang et al.~\cite{yes_backdoor} focused on attacking low-confidence ``edge cases'', which are hard to detect. However, our objective is to provide a general boosting method applicable to all source-target combinations.

\subsection{Defenses} \label{sec:denfense}
With FedAvg, an FL system can be easily corrupted by targeted poisoning attacks. Thus, a variety of defenses have been proposed to defend against poisoning attacks.

The Krum and Multi-Krum~\cite{KRUM} defenses select either one or multiple updates with the smallest distances to other updates as benign updates for the next communication round. The Median~\cite{median} defense computes the coordinate-wise median of local updates as the new update.
The Trimmed Mean~\cite{median} defense averages the remaining part of updates after removing the largest and smallest fractions. Bulyan\cite{Bulyan} combines Krum with Trimmed Mean. 
The aforementioned defenses are Byzantine-resilient and straightforward to implement, as they are either non-parametric or require minimal parameter tuning. Advanced poisoning attacks, such as those presented in~\cite{adv_len}, have demonstrated the ability to compromise these defenses by manipulating model weights. However, to circumvent Byzantine-resilient mechanisms, these attacks require a trade-off that significantly reduces their success rates.
Contra~\cite{contra} identifies suspicious models that display higher similarities with other local models, and then adaptively reduces their learning rates and likelihood of being selected for the next training round.
FLtrust\cite{fltrust} first computes the cosine similarity between the global update and local updates as the trust scores, and then uses the average of normalized local updates weighted by trust scores as the global update. Flare\cite{flare} designs the trust score based on the distance of penultimate layer parameters.
FLDetector~\cite{FLdetector} designs a malicious client detection method based on update consistency, which is implemented prior to secure aggregations.
Contra, FLTrust, FLARE, and FLDetector are effective in defense across various attack scenarios. However, these methods rely on different assumptions about the FL system and often require complex parameter tuning. For example, Contra heavily depends on detecting alignment (cosine similarity) among the updates from malicious clients. If the malicious local models exhibit diverse patterns, the effectiveness of this approach may be significantly reduced.

In addition, recent studies have demonstrated that differential privacy (DP)~\cite{DP_backdoor2,DP_backdoor1,uncovering}, implemented through clipping and adding noise to local model updates, can effectively enhance the FL model's robustness against poisoning attacks.
Among these studies, Flame~\cite{flame} stands out as a particularly noteworthy approach. It constrains the amount of injected DP noise by employing three defense phases: filtering deviated model updates, clipping scaled-up updates, and finally adding suitable noise to the updates.
In this paper, we include Flame as one of the baseline defense methods, as it has proven effective in mitigating the impact of poisoning attacks while preserving the model's performance on the main task.



\subsection{Design and Application of Soft Labels}
Traditionally, one-hot encoding, or hard label, has been widely used for training neural networks in classification tasks. However, hard labels fail to capture the shared underlying features across different classes. To address this problem, existing work has explored soft labels~\cite{revisit_soft, rethink_label_smoothing, soft_contrastive, soft_rethink_distill}, with a focus on enhancing the generalization and transferability of ML models by smoothing their decision boundaries.
Two prominent methods are commonly used to generate soft labels. The first is label smoothing\cite{rethink_label_smoothing}, which mixes hard labels with a uniform distribution over all classes. The second approach is knowledge distillation~\cite{distill_knowledge}, where the knowledge from a larger and pre-trained model (teacher model) is transferred to a smaller model (student model). During this process, the output probabilities of the teacher model serve as soft labels for the student model, preserving the relative information between classes. 

In this work, the attacker maliciously uses soft labels to mislabel intermediate classes, strategically influencing the decision boundary between the source and target classes in the latent feature space. This approach facilitates misclassification while preserving the stealthiness of the poisoning attacks.


\section{Threat Model}
\noindent
\textbf{Adversarial Goal.} The adversary's goal is to manipulate the FL model to cause samples from a specific source class, $c_\text{src}$,  to be misclassified as a designated target class, $c_\text{tgt}$.
Given a classification model $f_{\boldsymbol{w}}$, the attacker seeks to maximize the probability of incorrect classification:
\begin{equation}
    \max_{\boldsymbol{w}} \mathbb{P} \left( f_{\boldsymbol{w}}(\boldsymbol{x}) = c_\text{tgt} \mid y = c_\text{src} \right),
\end{equation}
where $\boldsymbol{x}$ represents an input sample from the source class $c_\text{src}$ and $y$ denotes its true label.

\vspace{0.2in}
\noindent 
\textbf{Adversarial Capabilities.} We assume that only the modification of existing data labels on malicious clients is required, without injecting additional malicious data samples or making extra modifications to model weights. BoTPA can be applied to both data poisoning and model poisoning attacks without disturbing their original workflow.
Additionally, backdoor attacks involving the attachment of triggers to images are not within the scope of this study, as we do not modify data features. 
In terms of the percentage of compromised clients, we assume that the attacker can compromise fewer than 33\% of the clients, adhering to the Byzantine fault tolerance model, which better represents real-world adversarial conditions. The more details on the target FL system and adversarial settings are given in  Section~\ref{sec:exp_setup}.

\noindent
\textbf{Defenses.} 
We assume that established defenses against poisoning attacks, such as Byzantine-resilient aggregation techniques and DP-based mechanisms, are in place. However, this study does not consider adaptive defenses and we assume that the defender has no prior knowledge of the details of the proposed attack.

\section{BoTPA Framework} \label{sec:strategy}
In this section, we start by presenting the design overview of the BoTPA framework, including the high-level step-by-step execution. Next, we explain in detail the construction of the Amplifier set and then the design of soft labels.  Finally, we provide a model analysis to establish a theoretical understanding of the BoTPA framework.

\begin{figure}[t] 
\centerline{\includegraphics[width=1\linewidth, trim =1.4cm 9.2cm 8.7cm 0cm, clip]{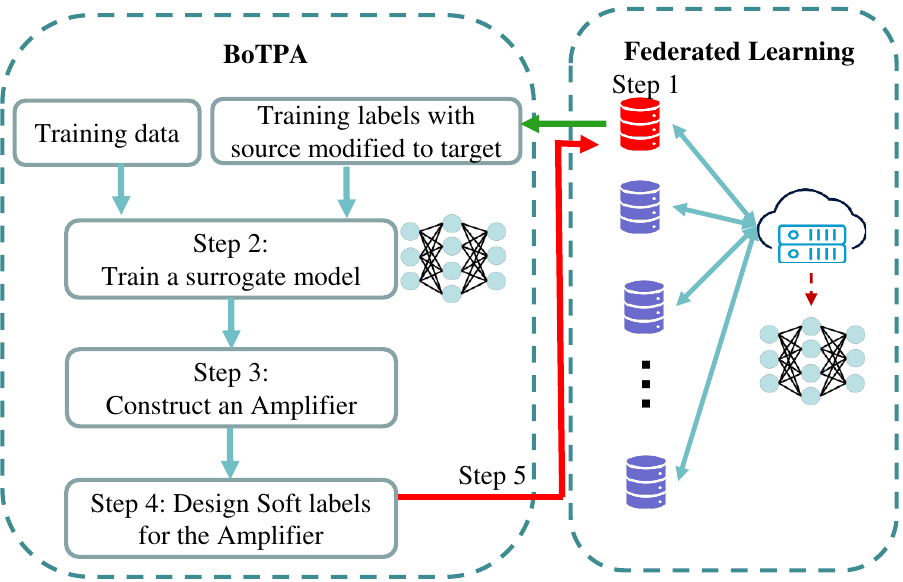}}
\caption{Overview of BoTPA. The right box represents the poisoned FL system, and the left box displays the BoTPA procedures.
}
\label{fig:Diagram}
\end{figure}

\subsection{System Design Overview}
During both data poisoning and model poisoning attacks,
the attacker compromises a small subset of clients in the FL system to manipulate the training dataset or the model. 
Although the capabilities of attackers differ between data poisoning and model poisoning attacks, both types of attacks begin by modifying data labels. In these vanilla attacks, only the data labels of the source class are modified, without considering the indirect impact of other classes regarding targeted misclassification.
To fill this gap, we propose the boosting method, BoTPA, which focuses on label manipulation in the pre-training stage.  
 The BoTPA overview is illustrated in Figure~\ref{fig:Diagram}, with the steps outlined below:
\begin{itemize}
    \item {Step 1:} In malicious clients, the attacker changes the data sample labels from the source class to the target class.
    \item {Step 2:} The attacker trains a centralized surrogate model using the contaminated local datasets.
    \item {Step 3:} The Amplifier set is constructed using selected intermediate classes that produce model updates similar to those of the source class, as detailed in Algorithm~\ref{alg:select}.
    \item {Step 4:} {Soft labels are generated for data samples in the Amplifier set} based on input similarity derived from latent feature distributions, as detailed in Algorithm~\ref{alg:design}.
    \item {Step 5:} In the datasets of malicious clients, the labels of samples in the Amplifier set are changed to the crafted soft labels.
\end{itemize}
In practice, the attacker designs the surrogate model based on its knowledge of the FL model architecture: (1) with access, the surrogate model has the same architecture as the FL model 
(see Sections~\ref{boost_data_poison} to~\ref{boost_noniid})
; (2) without access, the attacker customizes surrogate models with different architectures (see Section~\ref{sec:different_arch}). Next, Section \ref{sec:choose} and  Section \ref{sec:design_soft_labels} will explain Step 3 and Step 4 in more detail, respectively.

\begin{algorithm}[t]
\caption{Construction of Amplifier Set}\label{alg:select}
\begin{flushleft}
    \textbf{Input:} the dataset $\{S_1, S_2, \cdots, S_{N_C}\}$ in malicious clients; the pretrained surrogate model with weights $\boldsymbol{w}$; the predefined number of intermediate classes $N$; source class $c_{\text{src}}$; target class $c_{\text{tgt}}$. \\ 
    \textbf{Output:} the Amplifier set.
 \end{flushleft}
\begin{algorithmic}[1]

\For {$c_i \in [N_c] \setminus [c_\text{src}, c_{\text{tgt}}]$}  
    \For {$\boldsymbol{x}$ in $S_{c_i}$}
        \For {$\boldsymbol{x}^{\prime}$ in $S_{c_\text{src}}$}
        \State $IS_{\text{contrib}}(\boldsymbol{x}, \boldsymbol{x}^{\prime}) \gets \frac{\nabla_{\boldsymbol{w}} f_{\boldsymbol{w}}\left(\boldsymbol{x}\right)}{\left\| \nabla_{\boldsymbol{w}} f_{\boldsymbol{w}}\left(\boldsymbol{x}\right)\right\|}  \frac{\nabla_{\boldsymbol{w}} f_{\boldsymbol{w}}\left(\boldsymbol{x}^{\prime}\right)}{\left\| \nabla_{\boldsymbol{w}} f_{\boldsymbol{w}}\left(\boldsymbol{x}^{\prime}\right) \right\|}$
        \EndFor
    \EndFor
    \State Compute $CS_{\text{contrib}}(c_\text{src}, c_i)$ according to Eq.(\ref{eq:CS_contrib}).
\EndFor
\State $C \gets$ first $N$ classes in argsort$(\{CS_{\text{contrib}}(c_\text{src}, c_{i }\in [N_c] \setminus [c_\text{src}, c_{\text{tgt}}]\})$. 
\State $\text{Amplifier set}=\bigcup_{c_i \in C} S_{c_i}$ \\
\Return{\text{Amplifier} set}
\end{algorithmic}
\end{algorithm}

\subsection{Construction of Amplifier Set} \label{sec:choose}
To boost targeted misclassification,  we define a metric named ``Input Similarity from Contribution Degree'' to identify the classes that produce model updates similar to those of the source class during model training. Then, we can group data samples from these classes to form the Amplifier set. By mislabeling the Amplifier set, the attacker is able to intensify the malicious impact on the source data. 

\textbf{Input Similarity from Contribution Degrees.} For neural networks, two visually similar training images may have a very distinct influence on the training process, including the directions and magnitudes of model updates. We refer to the impact of input samples on the model updates as their contribution degrees to the training process.
In order to evaluate these contribution degrees, the input similarity, $IS_\text{contrib}$, is defined based on the output difference upon a model update~\cite{input_similarity}.
Considering two inputs, $\boldsymbol{x}$ and $\boldsymbol{x}^\prime$, and the model weights $\boldsymbol{w}$. The metric $IS_\text{contrib}$ is derived through the following steps: (1) slightly modifying the output value for $\boldsymbol{x}$ by $\varepsilon$ through the addition of a minor update $\delta \boldsymbol{w}$ to the model weights; (2) observing the extent to which the output for $\boldsymbol{x}^\prime$ changes due to $\delta \boldsymbol{w}$; (3) using the difference in output changes as $IS_\text{contrib}$. In step (2), the change in the output for $\boldsymbol{x}^\prime$ is given by $\varepsilon \frac{\nabla_{\boldsymbol{w}} f_{\boldsymbol{w}}\left(\boldsymbol{x}^{\prime}\right) \cdot \nabla_{\boldsymbol{w}} f_{\boldsymbol{w}}(\boldsymbol{x})}{\left\|\nabla_{\boldsymbol{w}} f_{\boldsymbol{w}}(\boldsymbol{x})\right\|^2}$, where the kernel $\frac{\nabla_{\boldsymbol{w}} f_{\boldsymbol{w}}\left(\boldsymbol{x}^{\prime}\right) \cdot \nabla_{\boldsymbol{w}} f_{\boldsymbol{w}}(\boldsymbol{x})}{\left\|\nabla_{\boldsymbol{w}} f_{\boldsymbol{w}}(\boldsymbol{x})\right\|^2}$ represents the influence of $\boldsymbol{x}$ over $\boldsymbol{x}^{\prime}$.\footnote{For more details on the derivation, please refer to~\cite{input_similarity}.} 
Thus, the symmetric normalized input similarity between $\boldsymbol{x}$ and $\boldsymbol{x}^{\prime}$ from contribution degrees is defined as
\begin{equation}\label{eq:IS_contrib}
    IS_{\text{contrib}}(\boldsymbol{x}, \boldsymbol{x}^{\prime}) = 
    \frac{\nabla_{\boldsymbol{w}} f_{\boldsymbol{w}}\left(\boldsymbol{x}\right)}{\left\| \nabla_{\boldsymbol{w}} f_{\boldsymbol{w}}\left(\boldsymbol{x}\right)\right\|} \cdot \frac{\nabla_{\boldsymbol{w}} f_{\boldsymbol{w}}\left(\boldsymbol{x}^{\prime}\right)}{\left\| \nabla_{\boldsymbol{w}} f_{\boldsymbol{w}}\left(\boldsymbol{x}^{\prime}\right) \right\|}. 
\end{equation}
Inputs with higher values of $IS_\text{contrib}(\boldsymbol{x}, \boldsymbol{x}^{\prime})$ are considered to induce more similar changes to the model updates.
Considering that we use the cross-entropy loss function during the training process, $f_{\boldsymbol{w}}\left(\boldsymbol{x}\right)$ is represented as  $- \log f_{\boldsymbol{w}}(\boldsymbol{x})$ in the calculations of $IS_\text{contrib}$.

\textbf{Construction of Amplifier Set.} In Step 2 of the BoTPA overview, a surrogate model is trained based on the local datasets with the modified source labels. The weights of this model are used to calculate $IS_\text{contrib}$ between individual data samples. The similarities of data samples from different classes are then averaged to obtain the class similarity. Let the data accessible to the attacker be denoted as $\{S_1, S_2, \cdots, S_{N_C}\}$, where $S_i$ represents the set of samples from the $i$-th class, and $N_C$ is the number of classes. With the surrogate model $\boldsymbol{w}$, the class similarity between classes $c_1$ and $c_2$ in terms of contribution degrees is represented as 
\begin{equation}
    CS_{\text{contrib}}(c_1, c_2) = \frac{1}{n_{c_1} n_{c_2}} \sum^{n_{c_1}}_{i=1}\sum^{n_{c_2}}_{j=1} IS_{\text{contrib}}(S_{c_1,i}, S_{c_2,j}), \label{eq:CS_contrib}
\end{equation}
where $S_{c_1,i}$ is the $i$-th sample in $S_{c_1}$, $S_{c_2,j}$ is the $j$-th sample in $S_{c_2}$, and $n_{c_1}$ and $n_{c_2}$ are the sizes of classes ${c_1}$ and ${c_2}$, respectively.

We calculate $CS_{\text{contrib}}$ between the source class $c_{src}$ and all other classes, excluding the target class $c_{tgt}$. Classes with the highest $CS_{\text{contrib}}$ scores are designated as intermediate classes and collectively denoted as set $C$. 
The Amplifier set is constructed by obtaining existing data samples from each intermediate class, and it is presented as 
\begin{equation}
\text{Amplifier set}=\bigcup_{c_i \in C} S_{c_i}.
\end{equation}
We utilize the surrogate model weights from the middle epoch of training, prior to convergence, to select intermediate classes. This is because $CS_{\text{contrib}}$ relies on the derivative of the loss function with respect to model weights, which becomes negligible after model convergence. Additionally, this approach reduces the impact of random model weight initialization on the selection process, ensuring the stability of the selected intermediate classes. The effectiveness of this intermediate class selection approach will be theoretically analyzed based on weight divergence in Section~\ref{subsec:theory}.

It is worth noting that we construct the Amplifier set using intermediate classes rather than individual data samples due to the inherent uncertainty in ML training processes~\cite{uncertainty}. Utilizing intermediate classes helps stabilize the similarity scores between individual images from different classes.

\begin{algorithm}[t]
\caption{Soft Label Design for the Amplifier Set}\label{alg:design}
\begin{flushleft}
    \textbf{Input:} the dataset $\{S_1, S_2, \cdots, S_{N_C}\}$ in malicious clients; the converged surrogate model with weights $\boldsymbol{w_\textit{conv}}$; the set of intermediate classes $C$; source class $c_{\text{src}}$; target class $c_{\text{tgt}}.$  \\
    \textbf{Output:} the soft labels for intermediate classes.
 \end{flushleft}
\begin{algorithmic}[1]

\For {$c_i \in C$}  \
    \For {$\boldsymbol{x}$ in $S_{c_i}$}
        \For {$\boldsymbol{x}^{\prime}$ in $S_{c_\text{src}}$}
        \State $    IS_{\text{ftrs}}(\boldsymbol{x}, \boldsymbol{x}^{\prime}) \gets 
    \frac{L_{\boldsymbol{w_\textit{conv}}}(\boldsymbol{x}) \cdot L_{\boldsymbol{w_\textit{conv}}}(\boldsymbol{x}^{\prime})}{\left\|L_{\boldsymbol{w_\textit{conv}}}(\boldsymbol{x}) \right\| \left\|L_{\boldsymbol{w_\textit{conv}}}(\boldsymbol{x}^{\prime}) \right\|}.$
        \EndFor
    \EndFor
    \State Compute $CS_{\text{ftrs}}(c_\text{src}, c_i)$ according to Eq.(\ref{eq:CS_ftrs}).
    \If{ $CS_{\text{ftrs}}(c_\text{src}, c_i) > 0$}
        \State $\boldsymbol{s}_{c_i} \gets CS_{\text{ftrs}}(c_\text{src}, c_i)*\boldsymbol{e}_{c_{\text{tgt}}} + (1-CS_{\text{ftrs}}(c_\text{src}, c_i)) * \boldsymbol{e}_{c_i}.$
    \Else
        \State $\boldsymbol{s}_{c_i} \gets \boldsymbol{e}_{c_i}.$
    \EndIf
\EndFor\\
\Return{$\{s_{c_i \in C}\}$}

\end{algorithmic}
\end{algorithm}

\subsection{Soft Label Design for Amplifier Set} \label{sec:design_soft_labels}
When enhancing the targeted poisoning effect, it is crucial to minimize the impact on the classification performance of data samples from intermediate classes. To achieve this, we adopt soft labels for intermediate classes, making subtle changes that still effectively boost targeted poisoning attacks. Moreover, when training neural networks, hard labels assume that data from different classes are completely separated and distinct.
However, classes within the same dataset always share similar features, such as trucks and automobiles in CIFAR-10\cite{cifar10}. Thus, we define a metric to measure the similarity in knowledge learned by the ML model for different inputs. This metric is subsequently used for designing soft labels.

\textbf{Input Similarity from Latent Feature Distributions.} 
The input similarity, $IS_{\text{ftrs}}$, is based on feature distributions derived from the outputs of the logits layer.
The logits layer, being the last hidden layer before the activation layer, has been demonstrated to yield effective feature representations~\cite{logit}. 
Let $L_{\boldsymbol{w}}: \boldsymbol{x} \rightarrow L_{\boldsymbol{w}}(\boldsymbol{x})$ denote the mapping from the input layer to the logits layer. The similarity between the feature distributions of two data samples, $\boldsymbol{x}$ and $\boldsymbol{x}^\prime$, is represented as the cosine similarity between their logits layer representations. This is mathematically expressed as 
\begin{equation}
    IS_{\text{ftrs}}(\boldsymbol{x}, \boldsymbol{x}^{\prime}) = 
    \frac{L_{\boldsymbol{w}}(\boldsymbol{x}) \cdot L_{\boldsymbol{w}}(\boldsymbol{x}^{\prime})}{\left\|L_{\boldsymbol{w}}(\boldsymbol{x}) \right\| \left\|L_{\boldsymbol{w}}(\boldsymbol{x}^{\prime}) \right\|}, \label{eq:IS_ftrs}
\end{equation}
where $IS_{\text{ftrs}}(\boldsymbol{x}, \boldsymbol{x}^{\prime})$ is bounded by $[-1,1]$. 

\textbf{Soft Label Design.} To design appropriate soft labels for each class, it is essential to calculate the similarity between classes.
 Let the data available to the attacker be denoted as $\{S_1, S_2, \cdots, S_{N_C}\}$, and the converged surrogate model as $\boldsymbol{w}_\textit{conv}$. The class similarity between two classes, $c_1$ and $c_2$, based on feature distributions, can be represented as
\begin{equation}
    CS_{\text{ftrs}}(c_1, c_2) = \frac{1}{n_{c_1} n_{c_2}} \sum^{n_{c_1}}_{i=1}\sum^{n_{c_2}}_{j=1} IS_{\text{ftrs}}(S_{c_1,i}, S_{c_2,j}). \label{eq:CS_ftrs}
\end{equation}

Instead of using binary values (0s and 1s), soft labels employ fractional values to capture the nuanced relationships between different data classes. From Equation (\ref{eq:CS_ftrs}), we can derive the class similarity between the source class and each intermediate class. Intuitively, to achieve the boosting objective, the portion of features in intermediate classes that are shared with the source class is maliciously labeled as the target, while the remaining features are correctly labeled.
 This approach is inspired by the label smoothing method~\cite{rethink_label_smoothing}, which involves using a uniform distribution over all classes to design soft labels for regularization. 
 For example, consider an intermediate class with class ID $c_z$. The original hard label for $c_z$ is $\boldsymbol{e}_{c_z} = [0, 0, \cdots, 1, \cdots 0]^{T}$, where the $c_z$-th element is 1 and all the other elements are 0s. Here, $\boldsymbol{e}_\text{src}$ represents the source hard label, and $\boldsymbol{e}_\text{tgt}$ represents the target hard label. 
The new label for $c_z$ is denoted as
\begin{equation}
s_{c_z} =
\begin{cases} 
e_{c_z}, & \text{if } CS_{\text{ftrs}}(c_\text{src}, c_z) < 0, \\
\text{Comb}(c_{src}, c_z), & \text{otherwise},
\end{cases}
\end{equation}
where $\text{Comb}(c_\text{src}, c_z) = CS_{\text{ftrs}}(c_\text{src}, c_z) \cdot e_{c_\text{tgt}} + \big(1 - CS_{ftrs}(c_\text{src}, c_z)\big) \cdot e_{c_z}.$
This piecewise function is defined such that if $CS_{\text{ftrs}}(c_\text{src}, c_z)$ is smaller than 0, $c_z$ retains its original hard label. Otherwise, the label for $c_z$  is replaced with a soft label computed using the convex combination. In the final step of BoTPA, all data samples in the Amplifier set are modified to the crafted labels.



\subsection{Theoretical Analysis} \label{subsec:theory}

\begin{figure}[tbp]
\centerline{\includegraphics[width=\linewidth, trim =0cm 12.6cm 13.0cm 0cm, clip]{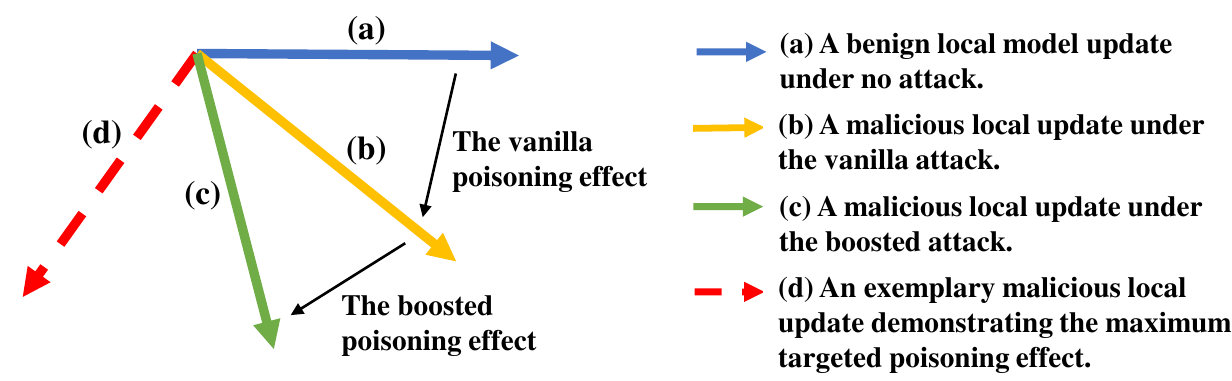}}
\caption{Illustration of local model updates under no attack, vanilla poisoning attack, and boosted poisoning attack.}
\label{fig:update_direction}
\end{figure}

To gain a deeper understanding of the boosting technique, we analyze the FL model under boosted data poisoning attacks, which can be conceptualized in two stages: (1) transitioning from no attack to a vanilla attack, and (2) transitioning from a vanilla attack to a boosted attack. 
While the actual attack does not proceed in these distinct stages, this abstraction facilitates a more straightforward analysis.
Figure~\ref{fig:update_direction} illustrates the local model updates, with step (a) to step (b) representing the first stage, and step (b) to step (c) representing the second stage.
Next, we provide a mathematical analysis based on the concept of weight divergence as described in~\cite{nonIID} during these two stages. The comparison of weight divergence demonstrates the impact of BoTPA on the model training process.

\subsubsection{Transition from no attack to vanilla attack}  

Supposing the cross-entropy loss is used for training the FL model, the loss function for the $k$-th benign client is
\begin{equation}  \label{eq:loss_fl}
\begin{aligned}
\ell(\boldsymbol{w}^{(k)})&
=-\sum_{i=1}^{N_C} p^{(k)}(y=i) \mathbb{E}_{\boldsymbol{x} \mid y=i}\left[\log f_{i}(\boldsymbol{x}, \boldsymbol{w}^{(k)})\right] ,
\end{aligned}
\end{equation}
where $f_{i}(\boldsymbol{x}, \boldsymbol{w})$ denotes the probability of classifying the input $x$ into class $i$, and $p^{(k)}(y=i)$ is the data distribution of class $i$. In contrast, if the $k$-th client is subjected to a data poisoning attack seeking to induce misclassfication from source class $s$ to target class $r$, the loss function is 
\begin{equation} \label{eq:loss_attack}
\begin{aligned}
\ell ^{\prime} (\boldsymbol{w}^{(k)})&
=\sum_{i=1, i\neq s}^{N_C}  \kappa_{i}(\boldsymbol{x}\mid y=i,  \boldsymbol{w}^{(k)})  
+ \kappa_{r}(\boldsymbol{x}\mid y=s,  \boldsymbol{w}^{(k)}),
\end{aligned}
\end{equation}
where the second term corresponds to the loss resulting from the misclassification and $\kappa(\cdot)$ is 
$$\kappa_{i}(\boldsymbol{x}\mid y=j,  \boldsymbol{w}) = - p^{(k)}(y=j) \mathbb{E}_{\boldsymbol{x} \mid y=j}\left[\log f_{i}(\boldsymbol{x}, \boldsymbol{w})\right].$$

\begin{proposition}
Suppose the FL model is updated using the gradient descent algorithm by minimizing the cross-entropy loss function. 
Let $m$ denote the number of local training iterations in each communication round.
If a targeted data poisoning attack, attempting to classify source class $s$ as target class $r$, occurs after $T$ communication rounds, the local weight divergence at time $t = mT + 1$ between the scenarios with and without the vanilla attack is given by
\begin{equation} \label{eq:6}
\begin{aligned}
    \Delta^{(k)}_{mT+1} &  =
    \boldsymbol{w}_{mT+1}^{\prime (k)} - \boldsymbol{w}_{mT+1}^{(k)} \\
    &= \eta p^{(k)}(y=s)  \mathbb{E}_{\boldsymbol{x} \mid y=s} [ \delta_{rs}(\boldsymbol{x}, \boldsymbol{w}^{ (k)}_{mT}) ],
\end{aligned}
\end{equation}
where $\delta_{rs}(\boldsymbol{x}, \boldsymbol{w}) =  \nabla_{\boldsymbol{w}} \log f_{r}(\boldsymbol{x}, \boldsymbol{w}) - \nabla_{\boldsymbol{w}} \log f_{s}(\boldsymbol{x}, \boldsymbol{w})  $, $\boldsymbol{w}^{\prime (k)}$ is the model weights under the vanilla attack in the $k$-th client, and $\eta$ is the learning rate. 
\end{proposition}
\begin{proof}
The loss functions employed by benign and malicious clients are represented by Equations (\ref{eq:loss_fl}) and (\ref{eq:loss_attack}), respectively. In each client, the local model is updated based on 
\begin{equation} \label{eq:sgd}
\boldsymbol{w}_{t}^{(k)}=\boldsymbol{w}_{t-1}^{(k)}-\eta \nabla_{\boldsymbol{w}} \ell\left(\boldsymbol{w}_{t-1}^{(k)}\right).
\end{equation}
Suppose the model weights without an attack in the $k$-th client at time $t=mT+1$ are denoted by $\boldsymbol{w}_{mT+1}^{(k)}$, and the model weights under the vanilla attack are denoted by $\boldsymbol{w}_{mT+1}^{\prime (k)}$. The weight divergence is 
\begin{equation} \label{eq:divergence_attack}
\begin{aligned}
    & \Delta^{(k)}_{mT+1} =
    \boldsymbol{w}_{mT+1}^{\prime (k)} - \boldsymbol{w}_{mT+1}^{(k)} \\
    &\stackrel{\text{put (\ref{eq:sgd}) in}} = \eta \left[  \nabla_{\boldsymbol{w}} \ell\left(\boldsymbol{w}_{mT}^{(k)}\right) -  \nabla_{\boldsymbol{w}} \ell ^{\prime}\left(\boldsymbol{w}_{mT}^{(k)}\right) \right] \\
    & \stackrel{\text{put (\ref{eq:loss_fl}) and (\ref{eq:loss_attack}) in}} = \eta p^{(k)}(y=s)  \mathbb{E}_{\boldsymbol{x} \mid y=s} [ \delta_{rs}(\boldsymbol{x}, \boldsymbol{w}^{ (k)}_{mT}) ], \\
\end{aligned}
\end{equation}
where 
\begin{equation}
\label{eq:delta_rs}
\delta_{rs}(\boldsymbol{x}, \boldsymbol{w}^{ (k)}_{mT}) =  \nabla_{\boldsymbol{w}} \log f_{r}(\boldsymbol{x}, \boldsymbol{w}^{(k)}_{mT}) - \nabla_{\boldsymbol{w}} \log f_{s}(\boldsymbol{x}, \boldsymbol{w}^{(k)}_{mT}). 
\end{equation}
\end{proof}

\subsubsection{Transition from vanilla attack to boosted attack}  
To implement the vanilla targeted poisoning attack, we introduce an intermediate class $z$ and include all data samples from this class to form the Amplifier set. By mislabeling this Amplifier set with the soft label $s_z = \lambda_z \boldsymbol{e}_r + (1-\lambda_z)\boldsymbol{e}_z$, the loss function is modified to 
\begin{equation} \label{eq:loss_boost}
\begin{split}
&\ell^{\prime \prime} (\boldsymbol{w}^{(k)})  \\
&=\sum_{i=1, i\neq s , i\neq z}^{N_C} \kappa_{i}(\boldsymbol{x}\mid y=i,  \boldsymbol{w}^{(k)}) +   \kappa_{r}(\boldsymbol{x}\mid y=s,  \boldsymbol{w}^{(k)}) \\
& + \lambda_z \kappa_{r}(\boldsymbol{x}\mid y=z, \boldsymbol{w}^{(k)}) + (1-\lambda_z)\kappa_{z}(\boldsymbol{x}\mid y=z, \boldsymbol{w}^{(k)}),
\end{split}
\end{equation}
where the last two terms represent the loss resulting from the data label modifications in the Amplifier set.

\begin{proposition}
Following Proposition 1, consider a boosted attack that mislabels the Amplifier set, which comprises all data samples from an intermediate class, with the soft label  $s_z = \lambda_z \boldsymbol{e}_r + (1-\lambda_z)\boldsymbol{e}_z$. In this scenario, if both the vanilla and boosted attacks take place after $T$ communication rounds, the local weight divergence at time $t=mT+1$ with and without BoTPA is given by
\begin{equation}
    \begin{aligned}
    \Delta^{\prime (k)}_{mT+1}& = 
    \boldsymbol{w}_{mT+1}^{\prime \prime (k)}-\boldsymbol{w}_{mT+1}^{\prime (k)} \\
        &= \lambda_z \eta p^{(k)}(y=z) \mathbb{E}_{\boldsymbol{x} \mid y=z} [   \delta_{rz} (\boldsymbol{x}, \boldsymbol{w}^{ (k)}_{mT}) ],
    \end{aligned}
\end{equation}
where $\boldsymbol{w}^{\prime \prime (k)}$ represents the model weights under the boosted attack in the $k$-th client. 
\end{proposition}
\begin{proof}
    Given the Amplifier set constructed using an intermediate class $z$ and its designed soft label $s_z$, the loss function for the boosted attack is represented by Equation~(\ref{eq:loss_boost}). 
    Suppose the model weights under the boosted attack at the time $t = mT+1$ are $\boldsymbol{w}_{mT+1}^{\prime \prime (k)}$. The weight divergence resulting from BoTPA, in addition to the vanilla attack, is then given by
\begin{equation} \label{eq:divergence_boost}
\begin{aligned}
    & \Delta^{\prime (k)}_{mT+1} =
    \boldsymbol{w}_{mT+1}^{\prime \prime (k)} - \boldsymbol{w}_{mT+1}^{\prime (k)} \\
    &\stackrel{\text{put (\ref{eq:sgd}) in}} = \eta \left[  \nabla_{\boldsymbol{w}} \ell ^{\prime} \left(\boldsymbol{w}_{mT}^{(k)}\right) -  \nabla_{\boldsymbol{w}} \ell ^{\prime\prime}\left(\boldsymbol{w}_{mT}^{(k)}\right) \right] \\
    & \stackrel{\text{put (\ref{eq:loss_attack}) and (\ref{eq:loss_boost}) in}} = \lambda_z \eta p^{(k)}(y=z) \mathbb{E}_{\boldsymbol{x} \mid y=z} [\delta_{rz} (\boldsymbol{x}, \boldsymbol{w}^{ (k)}_{mT})],
\end{aligned}
\end{equation}
where
\begin{equation} 
\label{eq:delta_rz}
\delta_{rz} (\boldsymbol{x}, \boldsymbol{w}^{ (k)}_{mT}) = \nabla_{\boldsymbol{w}} \log f_{r}(\boldsymbol{x}, \boldsymbol{w}^{(k)}_{mT}) - \nabla_{\boldsymbol{w}} \log f_{z}(\boldsymbol{x}, \boldsymbol{w}^{(k)}_{mT}).
\end{equation}
\end{proof}

\textbf{Analysis of Intermediate Class Selection:}
From Propositions 1 and 2, the vanilla and boosted poisoning effects steer the local model update in the directions of $ \delta_{rs}(\boldsymbol{x}, \boldsymbol{w}^{ (k)}_{mT})$  and $\delta_{rz} (\boldsymbol{x}, \boldsymbol{w}^{ (k)}_{mT})$ , respectively, as shown in Equations~(\ref{eq:delta_rs}) and~(\ref{eq:delta_rz}).
According to Equation (\ref{eq:IS_contrib}), the intermediate class $z$ is selected based on its similarity to the source class $s$ in contributing to the model update.
Additionally, in the actual computation, $- \log f_{\boldsymbol{w}}(\boldsymbol{x})$ is used to replace $f_{\boldsymbol{w}}(\boldsymbol{x})$ considering the cross-entropy loss function. 
Therefore, 
$\delta_{rz} (\boldsymbol{x}, \boldsymbol{w}^{ (k)}_{mT})$ closely aligns with $ \delta_{rs}(\boldsymbol{x}, \boldsymbol{w}^{ (k)}_{mT})$, signifying that both the vanilla and boosted poisoning effects consistently drive the malicious models to update in analogous vector directions towards the adversarial objective, as shown in Figure~\ref{fig:update_direction}. This explains why BoTPA can effectively enhance the impact of targeted poisoning attacks.

\textbf{Analysis of Soft Label Design:} In Equations~(\ref{eq:divergence_boost})
and
(\ref{eq:delta_rz}), the weight divergence is scaled by the factor $\lambda_z$, which originates from the soft label assigned to the intermediate class $z$. This indicates that, by utilizing soft labels, BoTPA strategically manipulates the feature representations of intermediate classes, positioning them as bridges between the source and target classes. This adjustment in the latent feature space reduces the separation of the decision boundary. Furthermore, compared to mislabeling intermediate classes via hard labels, the increased entropy in soft labels distributes the gradient contributions across the intermediate classes and the source class. This results in more subtle updates to the model parameters, thereby preserving the model's normal behaviors on the intermediate classes.

The boosting effect can accumulate when there are multiple training iterations locally within each communication round. Due to the enhanced poisoning effect in local models, the attacks can eventually influence the global model in favor of the adversarial objective.
Overall, the theoretical analysis of weight divergence underscores the critical role of accurately identifying and leveraging intermediate classes. This approach enables more effective manipulation of the model's decision boundaries, thereby amplifying the impact of the poisoning efforts.
This will be further demonstrated in Section~\ref{sec:feature_dis}.

\section{Experimental Setup} \label{sec:exp_setup}
In this section, we describe the datasets, FL system settings, evaluation metrics, and configurations for BoTPA.

\subsection{Datasets, FL systems and Adversarial Settings}
\textbf{Datasets.} We use three widely used datasets in the study of poisoning attacks~\cite{fltrust,flare}, Fashion-MNIST (FMNIST)\cite{FMNIST}, CIFAR-10 \cite{cifar10} and Colorectal Histology MNIST (CH-MNIST) \cite{chmnist}. FMNIST consists of ten classes of grey-scale fashion-related images (e.g., dresses, shirts). The size of each image is 28 $\times$ 28, with a total of 60,000 training images and 10,000 testing images. 
CIFAR-10 includes ten classes of colorful images. The dimensions of each image are $32 \times 32 \times 3$, with a total of 50,000 training images and 10,000 testing images. 
CH-MNIST consists of grey-scale images from eight different classes of human colorectal cancer, each sized at $64 \times 64$. It includes 5,000 images divided into a training set (4,000 images) and a testing set (1,000 images). 


\textbf{Data Distributions.} In IID scenarios (see Sections~\ref{boost_data_poison} and~\ref{sec:model_poisoning}), the dataset is uniformly and randomly assigned to each client. In non-IID scenarios (see Section~\ref{boost_noniid}), we adopt the setting from~\cite{FedNova, contra, FL_backdoor2} to simulate data imbalance using a Dirichlet distribution~\cite{dirichlet}, denoted as $\operatorname{Dir}_K(\beta)$. Here, $\beta$, ranging from 0 to positive infinity, represents the imbalance level, with smaller values indicating higher imbalance levels. 
In practice, $\beta$ is set to 0.5 and 1 to show BoTPA's performance in scenarios with different levels of data imbalance.\footnote{Please refer to~\cite{contra} for visualizations of the data distributions.}


\textbf{FL systems}:
The utilized FL system includes a global server and a set of local clients. 
Following the FL settings outlined in \cite{fltrust} and \cite{flare}, we configure 20 local clients for the FMNIST and CIFAR-10 datasets, and 10 local clients for the CH-MNIST dataset, considering its smaller size.
We adopt neural network architectures based on the design in~\cite{flare}, which integrate multiple Convolution + ReLU layers with Batch Normalization and Max Pooling, followed by Dense layers.  The hyperparameters for these models are detailed as follows. For the FMNIST dataset, the model consists of 2 convolutional layers with a kernel size of $5\times5$, and a hidden dense layer with 128 nodes. For the CIFAR-10 and CH-MNIST datasets, the models include 4 and 6 convolutional layers, respectively, with a kernel size of $3\times3$, and two hidden dense layers with 256 and 128 nodes. 
During the training process, each local model is updated locally for five epochs before being sent back to the server.
The total number of communication rounds between the server and clients is 25 for CIFAR-10 and FMNIST, and 50 for CH-MNIST. The Adam optimization algorithm is used with a learning rate of $10^{-3}$ for FMNIST and CIFAR-10, and $10^{-4}$ for CH-MNIST.

\textbf{Adversarial Settings:} 
Poisoning attacks in FL can be broadly categorized into two types: data poisoning and model poisoning. Model poisoning attacks generally achieve higher success rates and greater stealthiness, while data poisoning attacks are less effective due to their limited attack capabilities. BoTPA, as an attack boosting method operating during the pre-training stage, is adaptable to various attack scenarios. To comprehensively evaluate its effectiveness, we apply it to both data poisoning and model poisoning scenarios.
To determine the malicious client ratios, we follow the commonly adopted adversarial settings in existing attack and defense studies~\cite{FL_datapoisoning,flare}. Typically, the ratio of malicious clients starts at a low level. In this paper, we begin with 5\%, corresponding to one or two clients in the FL systems across different datasets. This small ratio is chosen to evaluate the attack’s effectiveness under strict attack conditions. Then, we gradually increase the ratio to 30\%, aligning with the Byzantine fault tolerance threshold of 33\%, to assess the attack’s performance across varying attack scenarios.

\subsection{Evaluation Metrics}
\textbf{Attack Success Rate (ASR)~\cite{fltrust, flare}:} 
We use ASR to evaluate the targeted misclassification rate, defined as the ratio of the number of source data samples misclassified as the target to the total number of source samples. The ASRs of vanilla attacks and boosted attacks utilizing BoTPA are referred to as V-ASR and B-ASR, respectively.

\textbf{Relative Increase in Attack Success Rate (RI-ASR):} 
This metric identifies the ASR increase caused by boosted attacks relative to vanilla attacks, and it is represented as $\text{RI-ASR} = \frac{\text{B-ASR}-\text{V-ASR}}{\text{V-ASR}}.$ RI-ASR is important for illustrating the performance of BoTPA in boosting the vanilla poisoning attacks. Throughout this paper, the terms ``RI-ASR" and its abbreviation ``RI" will be used interchangeably.


\subsection{BoTPA Configurations}
During the experimentation phase, the attacker constructs an Amplifier set using data samples from a fixed number of intermediate classes. Empirically, we set $N=2$ for FMNIST, $N=4$ for CIFAR-10, and $N=3$ for CH-MNIST, respectively.
To determine the value of N, we incrementally increase its value from 1 in an ``easy case'' until we observe a decline in RI-ASR, choosing the value that yields the highest RI-ASR. The term ``easy case'' refers to the attack scenario involving a source-target combination that achieves a high RI-ASR in boosting data poisoning attacks (see Sections~\ref{boost_data_poison} and~\ref{sec:model_poisoning} for more details). This operation aims to increase targeted misclassification and minimize the impact on other classes. The performance of targeted poisoning attacks varies significantly across different source-target combinations and percentages of malicious clients. Consequently, it is not practical to identify a single optimal N across all class combinations due to the partial order nature of the optimization. 
Considering the slight variations in class-specific accuracy across different training rounds after convergence, we evaluate CIFAR-10 and FMNIST by averaging the ASRs from the 16th to the 25th rounds over three runs, and CH-MNIST by averaging the ASRs from the 30th to the 50th rounds.

\begin{table*}[t]
    \centering
    \caption{Performance comparison between vanilla and boosted data poisoning attacks. ``Median case'' and ``best case'' refer to the source-target combinations with the median and highest RI-ASRs (\%), respectively, across all possible combinations. V-ASR (\%) and B-ASR (\%) refer to the corresponding vanilla and boosted ASRs of these selected cases. Note that the selected source-target combinations vary with different percentages of malicious clients.}
    \label{tab:boost_data_poison}
\newcolumntype{I}{!{\vrule width 1.0pt}}
\resizebox{2\columnwidth}{!}{
\begin{tabular}{c|c|c|c|c|c|c|c|c|c|c|c|c|c|c}
\hline
\multicolumn{1}{cI}{}   & \multicolumn{3}{cI}{5\%}                          & \multicolumn{3}{cI}{10\%}    & \multicolumn{3}{cI}{20\%}       &\multicolumn{3}{c}{30\%}                                            \\ \cline{2-13} 
\multicolumn{1}{cI}{}                           & \multicolumn{1}{c|}{V-ASR}  & \multicolumn{1}{c|}{B-ASR} & \multicolumn{1}{cI}{\textbf{RI}} 
& \multicolumn{1}{c|}{V-ASR}  & \multicolumn{1}{c|}{B-ASR} & \multicolumn{1}{cI}{\textbf{RI}} 
& \multicolumn{1}{c|}{V-ASR} & \multicolumn{1}{c|}{B-ASR} & \multicolumn{1}{cI}{\textbf{RI}}  
& \multicolumn{1}{c|}{V-ASR}  & \multicolumn{1}{c|}{B-ASR} & \multicolumn{1}{c}{\textbf{RI}}   
   \\ \hline \hline

\multicolumn{13}{c}{FMNIST}  \\ \hline
\multicolumn{1}{cI}{Median case}     
& \multicolumn{1}{c|}{10.2} & \multicolumn{1}{c|}{11.9} & \multicolumn{1}{cI}{\textbf{16.6}}  
& \multicolumn{1}{c|}{11.7} & \multicolumn{1}{c|}{14.7} & \multicolumn{1}{cI}{\textbf{26.2}}            
& \multicolumn{1}{c|}{22.3} & \multicolumn{1}{c|}{27.3} & \multicolumn{1}{cI}{\textbf{22.4}}     
& \multicolumn{1}{c|}{32.4} & \multicolumn{1}{c|}{40.6}  & \multicolumn{1}{c}{\textbf{25.3}}                   \\ \hline
\multicolumn{1}{cI}{Best case} 
& \multicolumn{1}{c|}{9.8} & \multicolumn{1}{c|}{12.1} & \multicolumn{1}{cI}{\textbf{23.4}}
& \multicolumn{1}{c|}{10.9} & \multicolumn{1}{c|}{14.3} & \multicolumn{1}{cI}{\textbf{30.7}}
& \multicolumn{1}{c|}{19.1} & \multicolumn{1}{c|}{24.9} & \multicolumn{1}{cI}{\textbf{30.8}} 
& \multicolumn{1}{c|}{4.9}   & \multicolumn{1}{c|}{11.4}     & \multicolumn{1}{c}{\textbf{130.9}}
                  \\ \hline  \hline
                  
\multicolumn{13}{c}{CIFAR-10}  \\ \hline
\multicolumn{1}{cI}{Median case}   
& \multicolumn{1}{c|}{13.2}             & \multicolumn{1}{c|}{15.2}   &  \multicolumn{1}{cI}{\textbf{15.3}} 
 & \multicolumn{1}{c|}{14.4}  & \multicolumn{1}{c|}{17.4}  & \multicolumn{1}{cI}{\textbf{21.5}}  
& \multicolumn{1}{c|}{19.4}      & \multicolumn{1}{c|}{27.5}       & \multicolumn{1}{cI}{\textbf{41.8}}   
& \multicolumn{1}{c|}{31.4}         & \multicolumn{1}{c|}{43.1}   & \multicolumn{1}{c}{\textbf{36.9}}
            \\ \hline
\multicolumn{1}{cI}{Best case} 
& \multicolumn{1}{c|}{8.8}  & \multicolumn{1}{c|}{10.7}    & \multicolumn{1}{cI}{\textbf{21.2}}
& \multicolumn{1}{c|}{9.2}  & \multicolumn{1}{c|}{12.4}    & \multicolumn{1}{cI}{\textbf{33.6}}
& \multicolumn{1}{c|}{7.1} & \multicolumn{1}{c|}{15.2} & \multicolumn{1}{cI}{\textbf{115.1}}              
& \multicolumn{1}{c|}{10.6} & \multicolumn{1}{c|}{21.5} & \multicolumn{1}{c}{\textbf{102.7}}             
    \\ \hline \hline

\multicolumn{13}{c}{CH-MNIST}  \\ \hline
\multicolumn{1}{cI}{Median case}  
& \multicolumn{1}{c|}{-}  & \multicolumn{1}{c|}{-} & \multicolumn{1}{cI}{-}  
& \multicolumn{1}{c|}{17.1}  & \multicolumn{1}{c|}{20.0}  & \multicolumn{1}{cI}{\textbf{17.2}} 
 & \multicolumn{1}{c|}{12.8}    & \multicolumn{1}{c|}{15.8}   & \multicolumn{1}{cI}{\textbf{23.5}}   
& \multicolumn{1}{c|}{20.0}   & \multicolumn{1}{c|}{26.6}  & \multicolumn{1}{c}{\textbf{32.9}}          
        \\ \hline
\multicolumn{1}{cI}{Best case} 
& \multicolumn{1}{c|}{-} & \multicolumn{1}{c|}{-}  & \multicolumn{1}{cI}{-}  
  & \multicolumn{1}{c|}{10.3}  & \multicolumn{1}{c|}{12.9} & \multicolumn{1}{cI}{\textbf{24.7}}
& \multicolumn{1}{c|}{10.9} & \multicolumn{1}{c|}{18.6} & \multicolumn{1}{cI}{\textbf{70.0}}     
& \multicolumn{1}{c|}{12.6}  & \multicolumn{1}{c|}{30.9}  & \multicolumn{1}{c}{\textbf{145.7}}      
\\ \hline
\end{tabular}
}
\end{table*}

\section{Evaluation} \label{sec:evaluation}
This section presents the performance evaluation of BoTPA on data and model poisoning attacks in various scenarios. 
Specifically, attacks are conducted in IID settings within Sections~\ref{boost_data_poison} and~\ref{sec:model_poisoning}, while Section~\ref{boost_noniid} focuses on the non-IID settings.
Next, we conduct ablation studies on various construction strategies for the Amplifier set and different architectures of the surrogate model. Finally, we visualize the feature distributions in the latent space to interpret how BoTPA affects the decision boundary between source and target classes.

\subsection{Targeted Data Poisoning Attack} \label{boost_data_poison}
In vanilla data poisoning attacks, the attacker changes source data labels to misclassify a specific source into an incorrect target. The attack performance varies with different source-target class combinations within a fixed modification budget. Therefore, we use the following two approaches to determine the source and target classes. 
    (1) \textbf{Median Case}: select the source and target with the \textit{median} RI-ASR.
    (2) \textbf{Best Case}: select the source and target with the \textit{highest} RI-ASR.
The median cases present the average boosting performance across all possible source-target combinations, whereas the best cases exemplify the utmost boosting effectiveness in terms of RI-ASR.
 While prior research focused on attacking cases with specific sources and targets, our paper presents a more generalized boosting approach.

 \begin{figure}[t] 
\centerline{\includegraphics[width=1\linewidth]{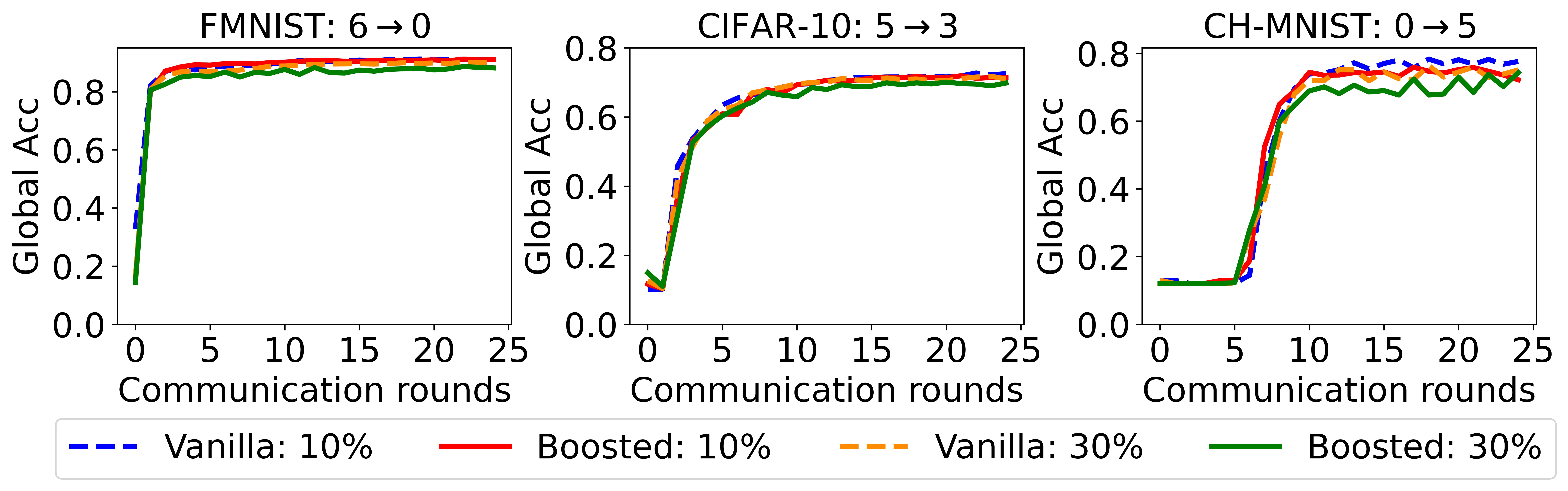}}
\caption{Global model accuracy under vanilla and boosted data poisoning attacks. $s \rightarrow r$ denotes an attack aiming to misclassify data samples from class $s$ to class $r$.
}
\label{fig:acc}
\end{figure}

 Table \ref{tab:boost_data_poison} presents the median and best cases under boosted data poisoning attacks, with the percentage of malicious clients ranging from 5\% to 30\%.
 While the median and best cases are selected based on RI-ASRs, their corresponding vanilla and boosted ASRs are also provided to showcase the attack capabilities in these selected cases. 
 Note that the selected source-target combinations differ in scenarios with different percentages of malicious clients; thus, boosted ASRs do not linearly increase with more malicious clients. 
Overall, in median cases, boosted attacks demonstrate RI-ASRs ranging from 15.3\% to 41.8\% across different percentages of malicious clients. This range of RI-ASRs is more representative as median cases consider all possible source-target combinations. Furthermore, in the best cases, boosted attacks achieve RI-ASRs of up to 145.7\%. In these best cases, the vanilla ASRs are consistently lower than those in median cases. Thus, the RI-ASRs illustrate the maximum boosting capability, particularly for source-target combinations that exhibit a weak vanilla poisoning effect.

 Another important observation from Table \ref{tab:boost_data_poison} is that the boosting effect does not necessarily correlate with the performance of vanilla attacks. For instance, in the case of CH-MNIST with 30\% malicious clients, the median case exhibits a higher ASR than the best case under the vanilla attack. However, the boosted ASR of the median case is comparatively weaker.
 This discrepancy arises because BoTPA primarily relies on the class similarity between the source class and intermediate classes. BoTPA can achieve a significant boosting impact if a cluster of intermediate classes shares a high degree of similarity with the source class. Conversely, the boosting technique is less effective if it lacks adequate intermediate classes with high similarity to the source class.

\begin{figure}[t]
\begin{minipage}[t]{1\linewidth}
    \includegraphics[width=\linewidth]{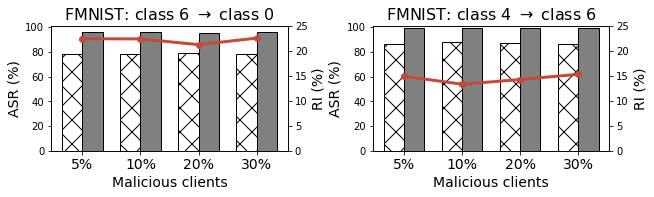}
\end{minipage}%
\\ 
\begin{minipage}[t]{1\linewidth}
    \includegraphics[width=\linewidth]{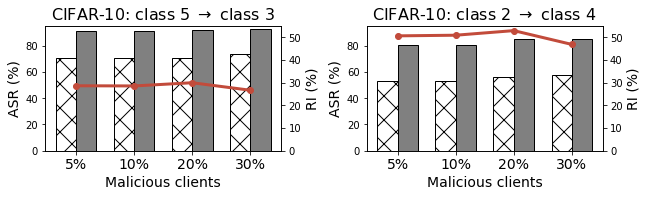}
\end{minipage} 
\\
\begin{minipage}[t]{1\linewidth}
    \includegraphics[width=\linewidth]{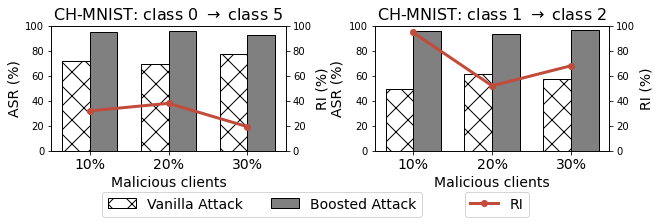}
\end{minipage} 
\caption{Performance comparison of stealthy  model poisoning attacks with varying ratios of malicious clients under Krum and Multi-Krum.}
\label{fig:krum}
\end{figure}

\textbf{Impact on Global Accuracy.} \label{sec:global_acc}
The objective of our attack is to misclassify solely the source class while minimizing the impact on other classes, including intermediate ones. By monitoring global accuracy, we can assess the overall effect on other classes.
 Figure~\ref{fig:acc} presents the accuracy of global models for various datasets under vanilla and boosted attacks. For clarity, only scenarios with 10\% and 30\% malicious clients are presented.
A detailed explanation of the selection strategy for the source-target combinations depicted in this figure will be provided in Section~\ref{sec:model_poisoning}.
The results indicate a negligible difference in global accuracy between the vanilla and boosted attacks, implying that the boosted attack can take place in a stealthy manner without significantly affecting the accuracy of other classes.

\subsection{Targeted Model Poisoning Attack} \label{sec:model_poisoning}

Model poisoning attacks provide more flexibility to evade defense mechanisms. 
To provide a comprehensive evaluation of stealthiness for targeted attacks, we employ mathematically proven byzantine-resilient defenses, including Krum \cite{KRUM}, Muti-Krum \cite{KRUM}, and Median \cite{median}, as well as a state-of-the-art defense, Flame~\cite{flame}.
The following describes the settings of model poisoning attacks under these defenses. 
    (1) \textit{Krum and Multi-Krum}: Attackers can break Krum by compromising a single client. Consequently, under Krum, there is only one malicious client, representing 5\% of the clients for FMNIST and CIFAR-10, and 10\% for CH-MNIST. In the case of Multi-Krum, attackers vary the percentage of malicious clients from 10\% to 30\%. Correspondingly, the number of local model updates selected by the defense also ranges from 10\% to 30\%.
    (2) \textit{Median and Flame}: Similarly, the percentage of malicious clients varies from 5\% to 30\%.

\begin{figure}[t]
\begin{minipage}[t]{1\linewidth}
    \includegraphics[width=\linewidth]{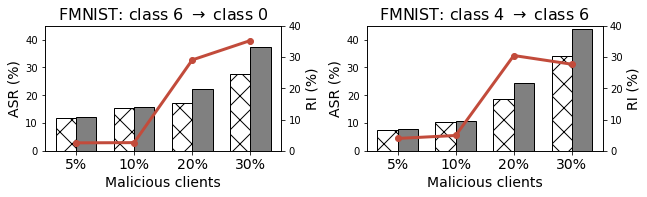}
\end{minipage}%
\\ 
\begin{minipage}[t]{1\linewidth}
    \includegraphics[width=\linewidth]{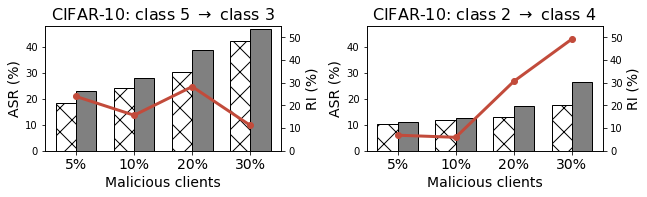}
\end{minipage} 
\\
\begin{minipage}[t]{1\linewidth}
    \includegraphics[width=\linewidth]{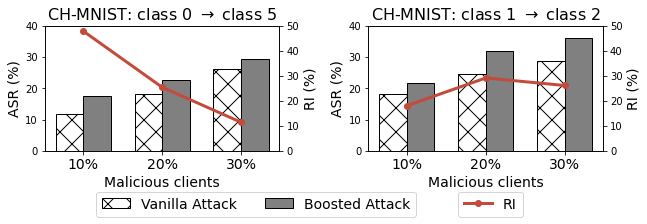}
\end{minipage} 
\caption{Performance comparison of stealthy model poisoning attacks under the Median defense.}
\label{fig:median}
\end{figure}

\begin{figure}[t]
\begin{minipage}[t]{1\linewidth}
    \includegraphics[width=\linewidth]{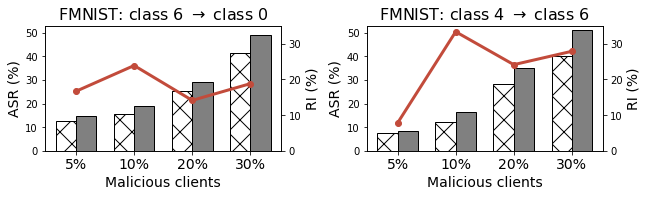}
\end{minipage}%
\\ 
\begin{minipage}[t]{1\linewidth}
    \includegraphics[width=\linewidth]{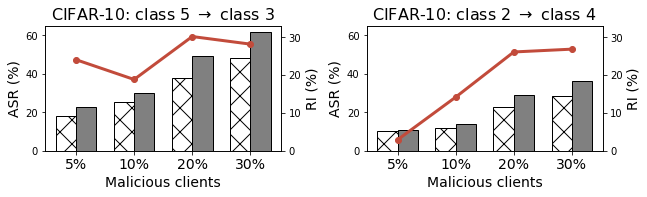}
\end{minipage} 
\\
\begin{minipage}[t]{1\linewidth}
    \includegraphics[width=\linewidth]{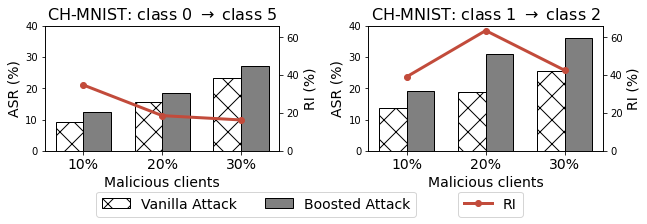}
\end{minipage} 
\caption{Performance comparison of stealthy  model poisoning attacks under the Flame defense.}
\label{fig:flame}
\end{figure}

For stealthy model poisoning attacks, the coefficients that balance the poisoning loss term with the stealthy loss terms must be adjusted according to the specific source and target classes~\cite{adv_len}.
 Therefore, following the source-target pairs evaluated in \cite{FL_datapoisoning}, we test two pairs of $c_\text{src}$ and $c_\text{tgt}$ for each dataset: one pair is selected from the best cases, while the other pair is chosen from the median cases in the data poisoning scenarios examined in Section~\ref{boost_data_poison}. These pairs represent scenarios with different levels of misclassification difficulty.
Specifically, for the best cases, we test class 6 (shirt) $\rightarrow$ class 0 (T-shirt) for FMNIST, class 5 (dog) $\rightarrow$ class 3 (cat) for CIFAR-10, and class 0 (tumor) $\rightarrow$ class 5 (mucosa) for CH-MNIST. For the median cases, we test class 4 (coat) $\rightarrow$ class 6 (shirt) for FMNIST, class 2 (bird) $\rightarrow$ class 4 (deer)  for CIFAR-10, and class 1 (stroma) $\rightarrow$ class 2 (complex) for CH-MNIST.

Figure \ref{fig:krum} presents the performance comparison between vanilla and boosted attacks against Krum and Multi-Krum.
We can see that BoTPA consistently and significantly boosts the poisoning performance across different datasets. The vanilla stealthy model poisoning attacks can bypass Krum and Multi-Krum defenses, but the poisoning capability is limited to keep the attack stealthy. With the labels in the Amplifier set being modified, the ASRs of boosted attacks reach nearly 100\%. This conclusively proves that the boosted attacks are more powerful, while still bypassing the defenses.

Similarly, Figures~\ref{fig:median} and~\ref{fig:flame} compare the performance of vanilla and boosted attacks against the defenses of Median and Flame, respectively. The results of RI-ASRs show that the boosting performance is not strictly correlated with the percentage of malicious clients. Under these defenses, RI-ASRs are influenced by the trade-off between stealthiness and poisoning impact, which varies depending on the source-target pair.
In most cases involving FMNIST and CIFAR-10, the highest RI-ASR is observed with 30\% malicious clients, achieving up to 49.2\% under Median and 33.3\% under Flame. 
However, in some cases, such as CH-MNIST (class 0 → class 5), the boosting performance degrades as the number of malicious clients increases. This degradation occurs for two reasons: (1) the vanilla ASRs with fewer malicious clients are lower, making it easier to achieve higher RI-ASRs; (2) a larger proportion of malicious clients makes the model changes caused by mislabeling the Amplifier set more detectable and, consequently, more likely to be excluded. Overall, BoTPA demonstrates robust boosting performance across various datasets, despite performance variations among different source-target combinations.

\begin{figure*}[htbp]
\begin{minipage}[t]{1\linewidth}
    \includegraphics[width=\linewidth]{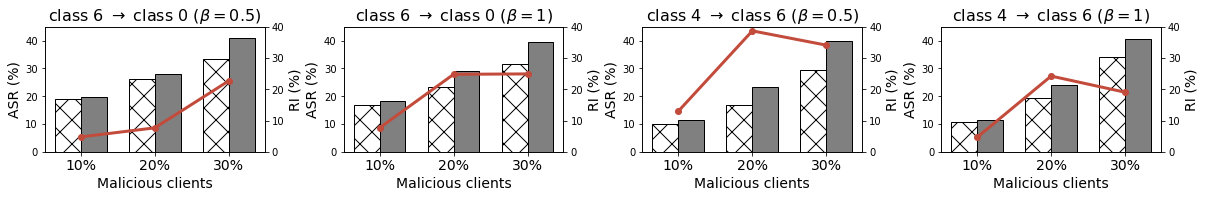}
    \subcaption{FMNIST}
\end{minipage}%
\\ 
\begin{minipage}[t]{1\linewidth}
    \includegraphics[width=\linewidth]{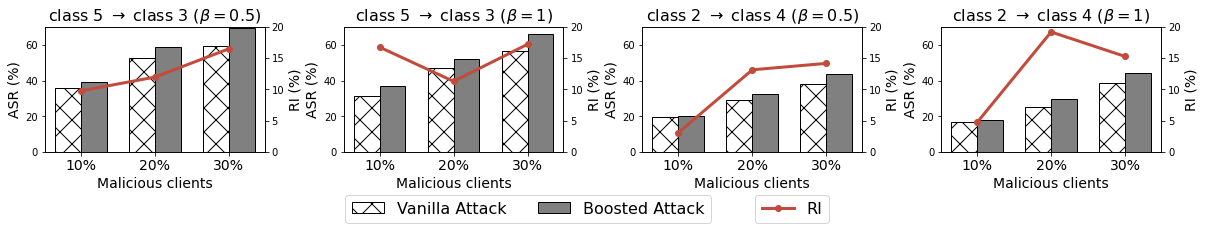}
    \subcaption{CIFAR-10}
\end{minipage} 
\caption{Performance comparison of data poisoning attacks in non-IID scenarios with different imbalance levels.}
\label{fig:non-iid}
\end{figure*}

\subsection{Performance on Non-IID data}
\label{boost_noniid}
To evaluate the effectiveness of BoTPA in non-IID scenarios, we perform data poisoning attacks with the malicious client percentage starting at 10\%, taking into account the class size imbalance among individual clients.
Since the number of data samples in the source class and intermediate classes significantly changes across different malicious clients, we run each case 20 times and use the average ARSs and RI-ASRs to demonstrate the boosting capability. 
To ensure a fair comparison, the same set of malicious clients is poisoned for both vanilla and boosted attacks in each run. However, across different runs, the malicious clients are randomly selected based on the specified percentage.

Figure~\ref{fig:non-iid} presents the performance of BoTPA on FMNIST and CIFAR-10 datasets. BoTPA achieves an RI-ASR up to 38.7\% across different malicious client ratios. Notably, the boosting performance remains relatively consistent for the same source-target combination across non-IID scenarios with varying data imbalances. The boosting effect is generally more pronounced with a higher percentage of malicious clients.
However, in some cases, the RI-ASRs for 20\% malicious clients are higher than those for 30\%. This is because RI-ASRs represent the relative increases, which are influenced by the baseline ASRs of the vanilla attacks.

\subsection{Ablation Study}
Here we conduct ablation experiments to evaluate the impact of different components in BoTPA. For this analysis, ASR values from vanilla attack scenarios are used as a baseline to clearly illustrate how different design choices affect the poisoning effect.

\subsubsection{Power of Intermediate Classes}
BoTPA enhances the poisoning effect by mislabeling the carefully selected data samples from intermediate classes. To demonstrate the effectiveness of the intermediate class selection, we construct the Amplifier set using two methods: (1) bringing data samples from intermediate classes; and (2) randomly selecting data samples from the other classes, instead of the source and target classes. We refer to the latter method as ``Random". We compare the boosting performance of the two methods in Figure~\ref{fig:random}, with each setting repeated 20 times. The Amplifier sets used in BoTPA and Random are of the same size and are mislabeled in the same manner as the soft label design. 
BoTPA exhibits superior boosting performance compared to Random. In many cases, such as CIFAR-10 (class 5 $\rightarrow$ class 3), the Random method even degrades the poisoning effect, instead of boosting it.

%

\begin{figure}[t]
\begin{minipage}[t]{1\linewidth}
    \includegraphics[width=\linewidth]{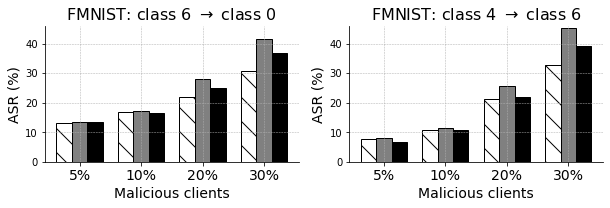}
\end{minipage}%
\\ 
\begin{minipage}[t]{1\linewidth}
    \includegraphics[width=\linewidth]{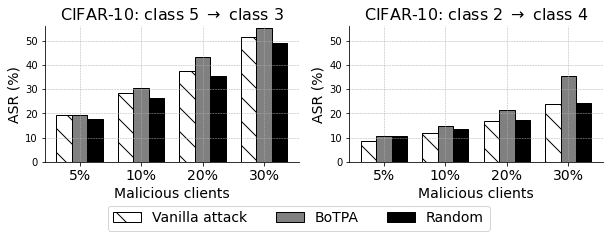}
\end{minipage} 
\caption{Performance comparison of boosted data poisoning attacks using different amplifier set construction methods. }
\label{fig:random}
\end{figure}

\subsubsection{Effect of Custom Surrogate Models} \label{sec:different_arch}
Here, we explore a weak attack scenario where attackers have access solely to the data but lack access to local models, including their architectures and parameters. This constraint limits the attackers to data poisoning and necessitates that they independently devise a custom surrogate model for implementing BoTPA.
To examine the effectiveness of the boosting technique under this black-box attack setting, we select two distinct surrogate models for each dataset. For FMNIST, we employ a simple model comprising one convolutional layer with batch normalization and two dense layers, as well as a ResNet20 model. For CIFAR-10, we utilize a LeNet-5 model and a ResNet20 model. These two custom surrogate models are referred to as ``Simple" and ``Deep", respectively.

In Figure \ref{fig:transfer}, we compare the performance of boosted data poisoning attacks using different surrogate model architectures. The term ``Identical" in the figure indicates that the surrogate model shares the same architecture as the actual local model.
The results show that the boosting technique remains effective under the black-box attack setting because of the knowledge transferability in neural networks~\cite{distill_knowledge, know_transferability}. The knowledge learned by a neural network can be transferred to the models with different architectures in the same task.  
Note that the boosted attacks using an identical surrogate model do not necessarily achieve the best performance. 
This phenomenon can be attributed to several factors. Firstly, features from different classes are entangled in the high-dimensional feature space, while some information is lost in logits layer representations. Secondly, our approach only considers the relationships among source, target, and intermediate classes for targeted poisoning, neglecting more complex relationships that may exist in the dataset. Thirdly, the inherent complexity of neural networks further contributes to variability in performance.

\subsection{Feature Distributions in Latent Space}\label{sec:feature_dis}
 To understand how BoTPA affects the decision boundary between the source and target classes, we employ a visualization scheme based on penultimate layer representations, as proposed in~\cite{pen_vis}. This approach allows us to visually observe changes in the feature space.
Figure~\ref{fig:visualization} illustrates the latent feature space under vanilla and boosted poisoning attacks. For clarity, only the source and target classes are shown. We observe that
(1) the source and target classes get closer in the latent feature space and (2) the clusters of source and target samples are tighter. This figure indicates that with BoTPA, data samples from the source class are more likely to cross the decision boundary and be misclassified as the target class.

\begin{figure}[t]
\begin{minipage}[t]{1\linewidth}
    \includegraphics[width=\linewidth]{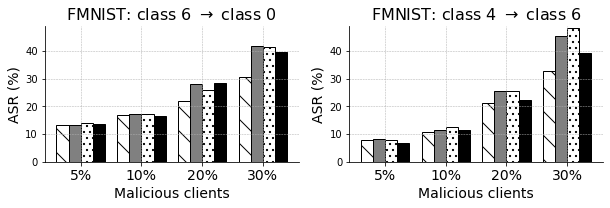}
\end{minipage}%
\\ 
\begin{minipage}[t]{1\linewidth}
    \includegraphics[width=\linewidth]{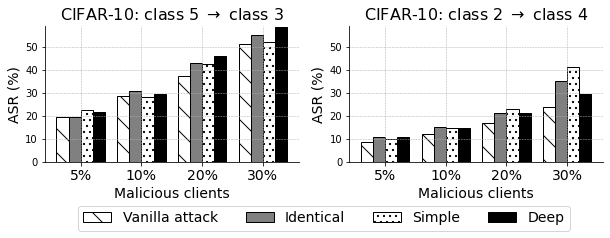}
\end{minipage} 
\caption{Performance comparison of boosted data poisoning attacks with various surrogate model architectures.}
\label{fig:transfer}
\end{figure}

\begin{figure}[t]
    \centering
    \subfloat[\centering class 6 $\rightarrow$ class 0]{{\includegraphics[scale=0.39]{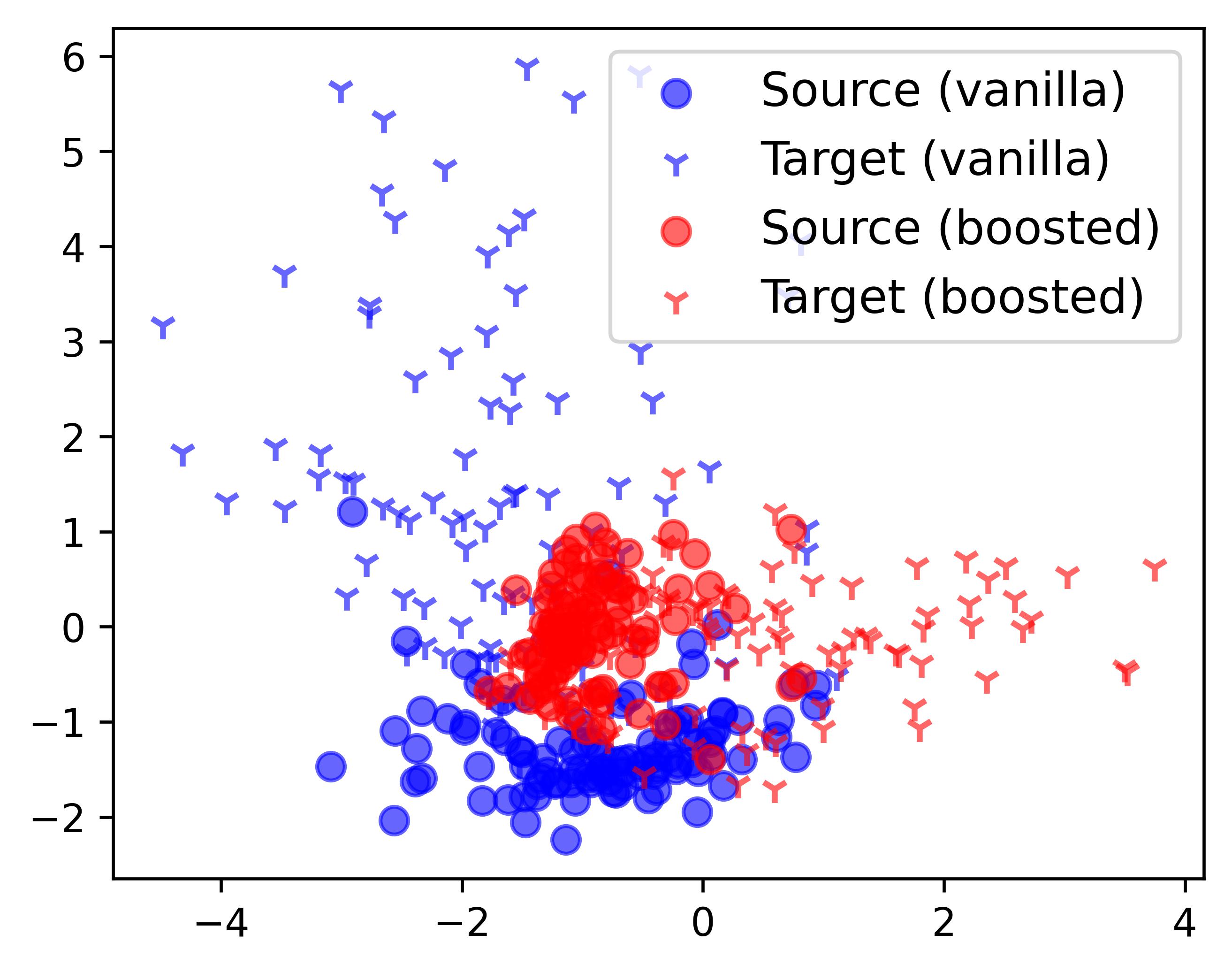} }}%
    \subfloat[\centering class 4 $\rightarrow$ class 6]{{\includegraphics[scale=0.39]{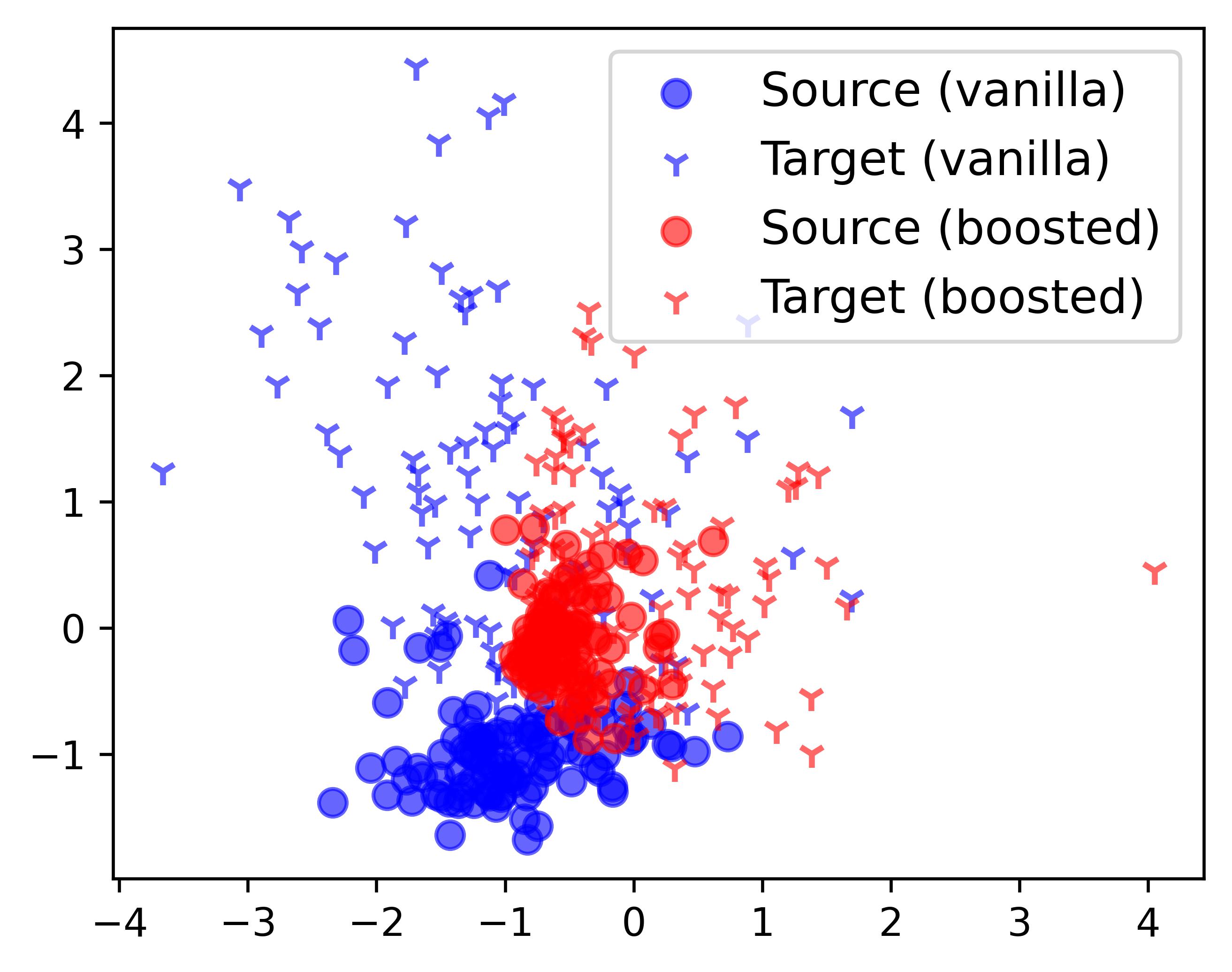} }}%
    \caption{Visualization of latent features for source and target classes in FMNIST under vanilla and boosted data poisoning attacks. 
    For example, ``source (vanilla)" denotes the feature representations of data samples in the source class under the vanilla attack.
    }
    \label{fig:visualization}%
\end{figure}

\section{Potential Defense}
Targeted poisoning attacks introduce more subtle changes to malicious updates compared to untargeted poisoning attacks. In this context, we explore a defense mechanism specifically designed to counter targeted poisoning attacks by examining the performance divergence between malicious and benign updates on \textit{individual classes}. Since our boosting technique relies on label modifications on the Amplifier set, a potential defense would involve assessing the performance divergence of different local models on individual classes to identify the possible intermediate classes and malicious clients.

Using the scenario of misclassifying class 4 to class 6 with FMNIST as an example, we compare the divergence between logits layer representations from malicious and benign clients given the same inputs, as illustrated in Figure~\ref{fig:defense}. In this figure, classes other than the source, target, and intermediate classes are designated as \textit{clean classes}. 
It is important to note the difference between Figure~\ref{fig:visualization} and Figure~\ref{fig:defense}. Figure~\ref{fig:visualization} aims to illustrate the changes in feature distributions of the source and target classes based on the global model, whereas Figure~\ref{fig:defense} aims to depict the differences in data representations between benign and malicious local models.

From Figure~\ref{fig:defense}, it is observed that the representations of the intermediate class from malicious and benign clients are visibly separated and dense. In contrast, the representations of the clean class from malicious and benign clients are indistinguishable and sparse. This observation suggests that the density divergence in logits layer representations between different local models is a promising indicator for detecting intermediate classes. If the representation densities of a class are consistent across local models, then the class is considered a clean class. Conversely, if there is significant density divergence, the class can be identified as an intermediate class.
Subsequently, clients in which intermediate classes are detected can be flagged as suspicious and removed from the model aggregation process.

\begin{figure}[t]%
    \centering
    \subfloat[\centering Intermediate class]{{\includegraphics[scale=0.045]{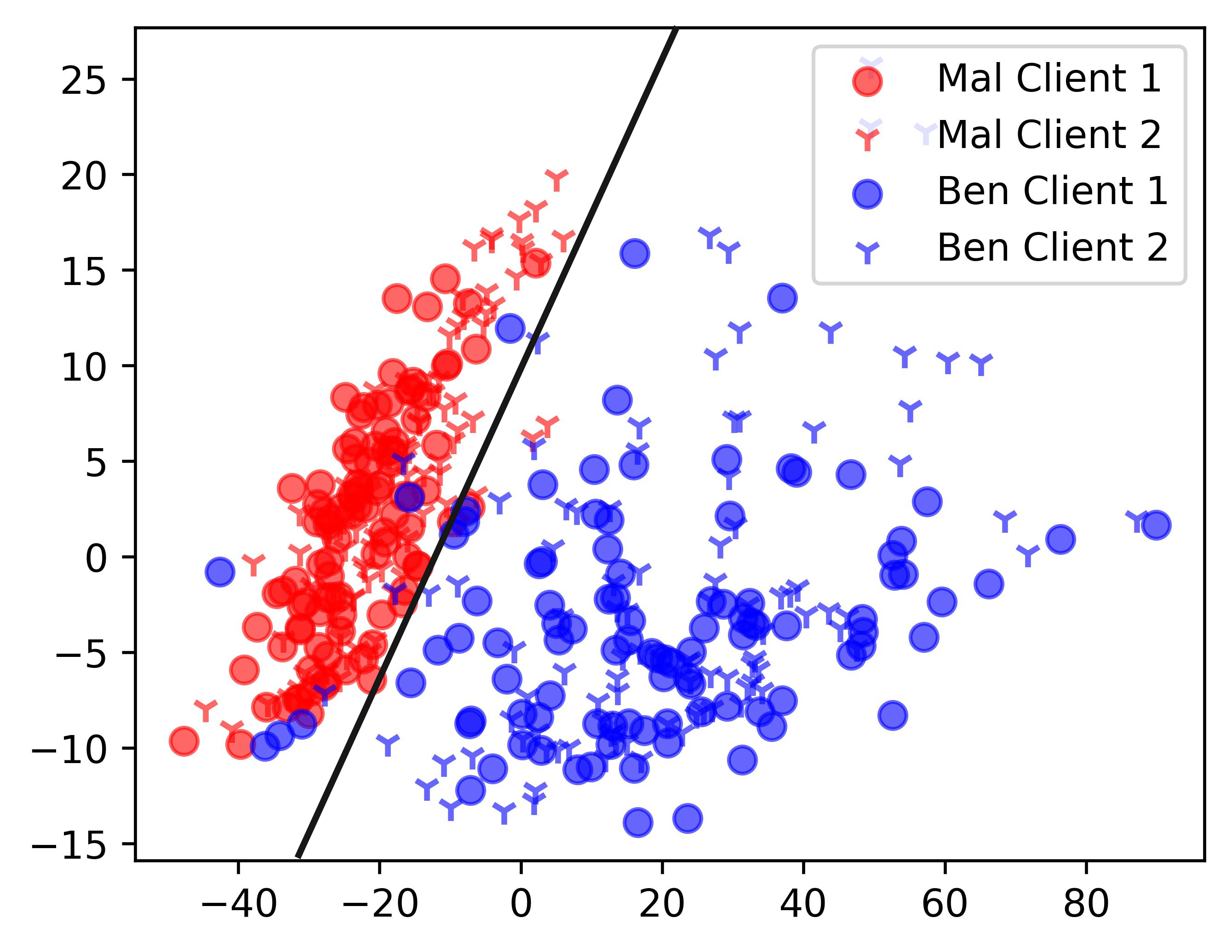} }}%
    \subfloat[\centering Clean class]{{\includegraphics[scale=0.375]{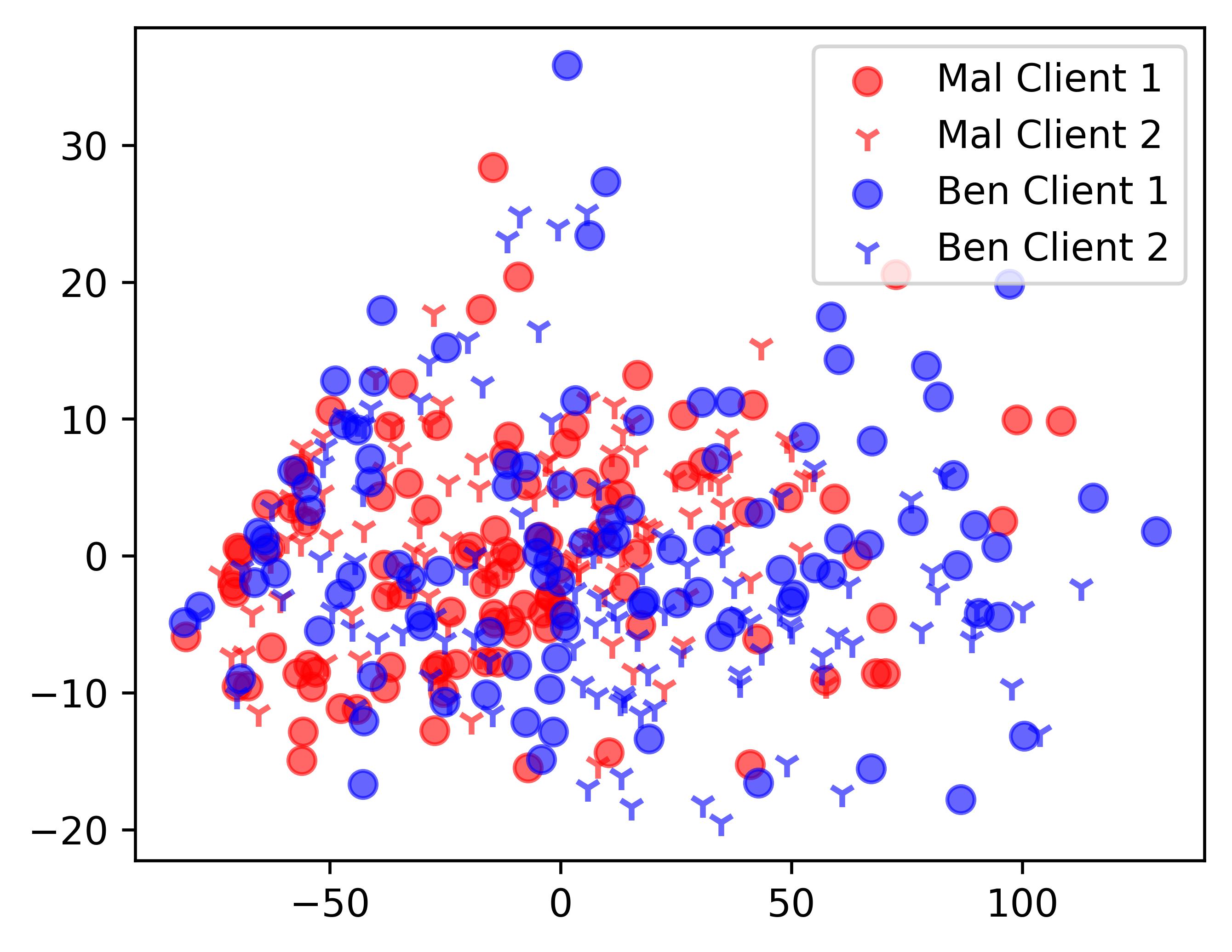} }}%
    \caption{Logits layer representations of intermediate and clean classes under boosted stealthy poisoning attack: data representations are extracted from two  malicious and two benign local models.
    }%
    \label{fig:defense}%
\end{figure}

\section{Discussion}
In this section, we analyze the computational complexity of BoTPA and discuss strategies to improve its scalability in FL systems with complex models or larger datasets. Finally, we explore its potential adaptation for real-world FL applications.

\textbf{Computational Complexity:} We assess computational complexity from three perspectives: training the surrogate model, selecting intermediate classes, and designing soft labels.
For training the surrogate model, the complexity depends primarily on the size of the model and the size of the data set. In this study, we assume that the attacker utilizes all the data from compromised clients, and so the computational cost is determined by the malicious-client ratio and the number of data samples from these clients. Regarding the architecture of the surrogate model, it is important to note that the performance of the boosted attacks does not necessarily improve with deeper models, as shown in Figure~\ref{fig:transfer}. A simpler and shallower model that achieves performance comparable to the target FL model is preferable to reduce the training time overhead.
For intermediate class selection and soft label design, the computational complexity in this step is affected by the number of intermediate classes, denoted as $\lvert C \rvert$. For a source class with $n_s$ samples and an intermediate class with $n_i$ samples, the complexity is proportional to $O(n_s \cdot n_i)$. Considering all intermediate classes, the total complexity becomes $O(n_s \cdot \sum_{i=1}^{\lvert C \rvert} n_i)$.

\textbf{Scalability:} To improve scalability in scenarios involving complex models and larger datasets, several strategies can be employed. First, instead of utilizing all samples from the source and intermediate classes, a random subset of samples can be selected to reduce computational overhead. From the perspective of model architecture, adopting a simpler surrogate model can significantly decrease the training time. Even if the surrogate model's architecture differs slightly from the target model, it can still effectively capture meaningful class relationships. In addition, early stopping can be implemented to stop the training of the surrogate model once it achieves a predefined accuracy threshold, eliminating the need to wait for full convergence.

\textbf{Attacking real-world FL systems:} 
Designing effective attacks in FL systems for real-world domains, such as healthcare~\cite{healthcare}, presents several significant challenges. One critical issue is dealing with imbalanced datasets, which include class imbalance, biases in the data, and data sparsity. 
In real-world scenarios, implementing BoTPA requires adaptation to account for various types of data biases. When dealing with quantity and label distribution biases, the design of the Amplifier must consider differences in class size to ensure effective manipulation. For feature bias, individual data samples should be utilized in constructing the Amplifier, rather than grouping them by class, to better capture nuanced variations within the dataset. 

\section{Conclusion \& Future Work}
In this paper, we propose BoTPA, a generalized framework designed to enhance the attack success rates (ASRs) of various targeted poisoning attacks, including both data and model poisoning attacks, by only requiring modifications to existing data labels. The core innovation of BoTPA lies in leveraging intermediate classes (i.e., classes other than the source and target classes) to increase the probability of specific misclassifications.
Our study shows that BoTPA effectively breaches the boundary between source and target classes by manipulating feature distributions in the latent space. We conduct comprehensive evaluations across three datasets, assessing the average boosting capability across all possible source-target combinations in data poisoning attacks. Furthermore, in the context of model poisoning attacks, BoTPA significantly enhances the poisoning impact while maintaining stealthiness, even in the presence of well-established defense mechanisms.

In future research, it is essential to gain a deeper understanding of the factors influencing the differences in ASRs among various source-target combinations. Moreover, developing a robust density-based defense mechanism is crucial. Such a defense could 
effectively identify subtle malicious modifications to data, thereby enhancing the security and resilience of federated learning systems.

\section*{Acknowledgments}
We would like to thank the anonymous reviewers
for their constructive comments, which help improve the quality of this article. This work was supported in part by the
US Office of Naval Research under Grant N00014-23-1-2158 and the National Science Foundation under
Grant CNS-2317829.


\bibliographystyle{IEEEtran}
\bibliography{main}


\section{Biography Section}

\vspace{-0.4in}

\begin{IEEEbiography}
[{\includegraphics[width=1in,height=1.25in,clip, trim=0mm 8mm 0mm 7mm, keepaspectratio]{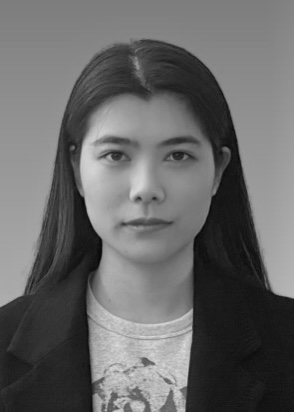}}]%
{Shihua Sun} received a B.Sc. degree in Electrical Science and Technology from the University of Electronic Science and Technology of China in 2018 and an M.Sc. degree in Electrical and Computer Engineering from the University of California, San Diego, USA, in 2020. She is currently pursuing a Ph.D. degree with the Department of Electrical and Computer Engineering at Virginia Tech, USA. Her research interests include adversarial machine learning and network security.
\end{IEEEbiography}

\vspace{-0.55in}

\begin{IEEEbiography}
[{\includegraphics[width=1in,height=1.25in,clip, trim=30mm 25mm 30mm 20mm, keepaspectratio]{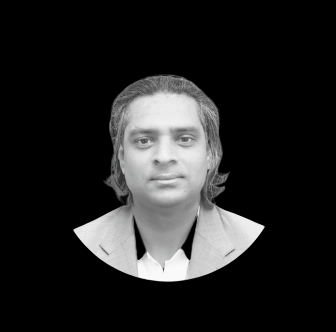}}]%
{Shridatt Sugrim} received a Ph.D. degree in Electrical and Computer Engineering from Rutgers University in 2020. He is currently a senior research scientist at Kryptowire Labs. His research interests cover cyber security, machine learning, data networks, probabilistic models and model evaluation.
\end{IEEEbiography}

\vspace{-0.45in}
\begin{IEEEbiography}
[{\includegraphics[width=1in,height=1.25in,clip,keepaspectratio]{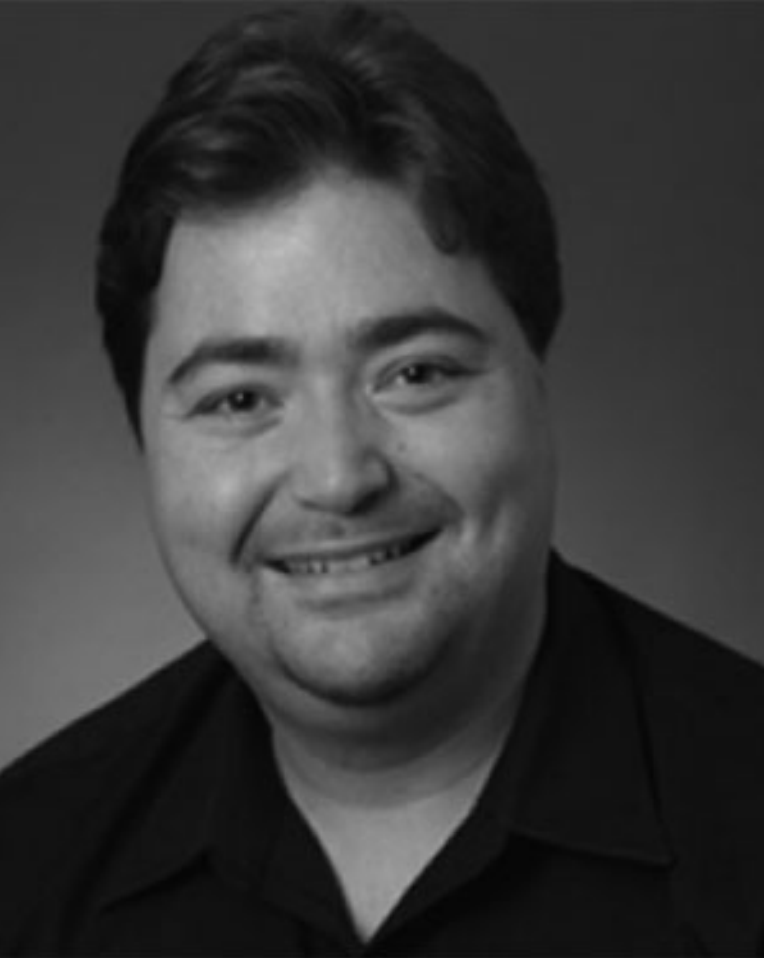}}]%
{Angelos Stavrou}
received a Ph.D. degree in computer science from Columbia University in 2007. He is currently a professor in the Department of Electrical and Computer Engineering at Virginia Tech, USA. His research interests include large systems security \& survivability, intrusion detection systems, privacy \& anonymity, and security for MANETs and mobile devices.
\end{IEEEbiography}

\vspace{-0.45in}
\begin{IEEEbiography}
[{\includegraphics[width=1in,height=1.25in,clip,keepaspectratio]{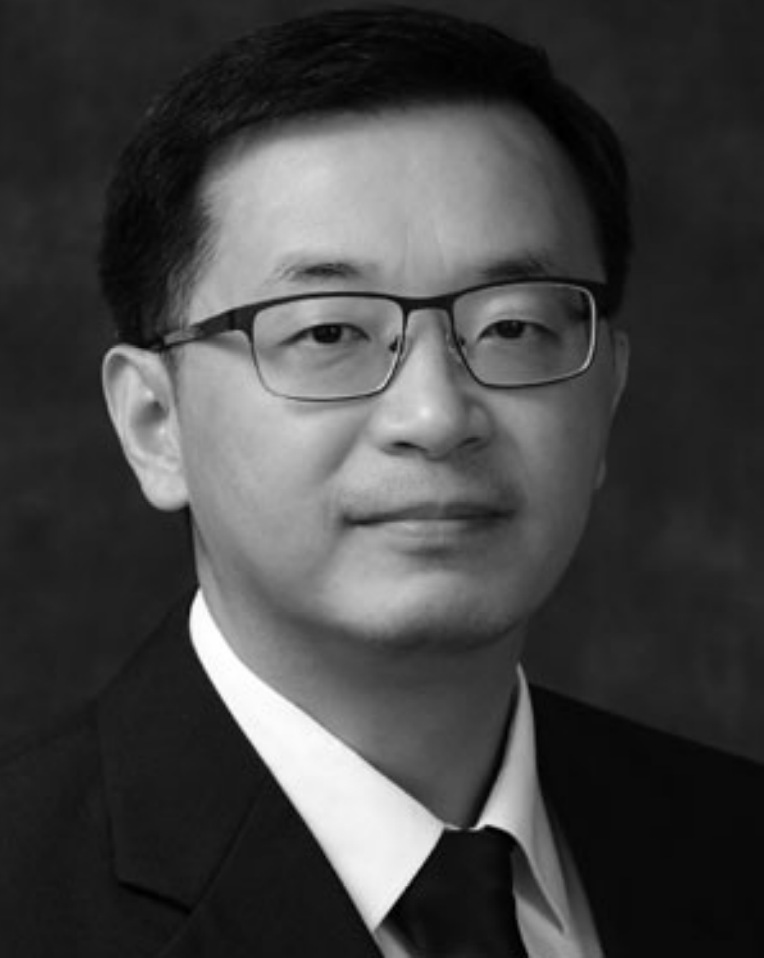}}]%
{Haining Wang}
(Fellow, IEEE) received a Ph.D. degree in computer science and engineering from the University of Michigan, Ann Arbor, MI, USA, in 2003. He is currently a professor in the Department of Electrical and Computer Engineering at Virginia Tech, USA. His research interests include security, networking systems, cloud computing, and cyber-physical systems.
\end{IEEEbiography}

\vfill

\end{document}